\documentclass[11pt]{article}
\usepackage{amsmath}
\usepackage{amsfonts}
\usepackage{subcaption}
\usepackage{graphicx}
\usepackage{latexsym}
\usepackage{url}
\usepackage{color} 
\usepackage{dsfont}
\usepackage{amssymb}
\usepackage{amsthm}
\usepackage{comment}
\usepackage{enumerate}

\excludecomment{en-text}
\includecomment{jp-text}
\excludecomment{comment}

\newcommand{\R}{\mathbb{R}}

\newcommand{\E}{\mathbb{E}}
\newcommand{\PP}{\mathbb{P}}
\newcommand{\N}{\mathbb{N}}
\newcommand{\Z}{\mathbb{Z}}

\newcommand{\toop}{\stackrel{\PP}{\longrightarrow}}

\newcommand{\schw}{\stackrel{d}{\longrightarrow}}
\newcommand{\eqschw}{\stackrel{d}{=}}

\newcommand{\toas}{\stackrel{\mbox{\tiny a.s.}}{\longrightarrow}}

\newcommand{\bee}{\begin{equation}}
\newcommand{\eee}{\end{equation}}
\newcommand{\bea}{\begin{eqnarray}}
\newcommand{\eea}{\end{eqnarray}}
\newcommand{\bean}{\begin{eqnarray*}}
\newcommand{\eean}{\end{eqnarray*}}

\renewcommand{\theequation}{\arabic{section}.\arabic{equation}}

\newtheorem{prop}{Proposition}[section]

\newtheorem{lem}[prop]{Lemma}

\newtheorem{theo}[prop]{Theorem}
\newtheorem{rem}[prop]{Remark}

\begin{document}

\title{Estimation of the linear  fractional  stable motion}
\author{Stepan Mazur \thanks{School of Business, \"Orebro University,  Fakultetsgatan 1, 70281  \"Orebro, Sweden, Email: stepan.mazur@oru.se.} \and 
Dmitry Otryakhin \thanks{Department of Mathematics, Aarhus University,  Ny Munkegade 118, 8000 Aarhus,
Denmark, Email: d.otryakhin@math.au.dk.} \and 
Mark Podolskij  \thanks{Department of Mathematics, Aarhus University,  Ny Munkegade 118, 8000 Aarhus,
Denmark, Email: mpodolskij@math.au.dk.} 
}

\maketitle

\begin{abstract}
In this paper we investigate the parametric inference for the linear fractional stable motion in high and low frequency setting. The symmetric linear fractional stable motion is a three-parameter family, which  
constitutes a natural non-Gaussian analogue of the scaled fractional Brownian motion. It is fully characterised by the scaling parameter $\sigma>0$, the self-similarity parameter $H \in (0,1)$ and the stability index $\alpha \in (0,2)$ of the driving stable motion. The parametric estimation of the model
is inspired by the limit theory for stationary increments L\'evy moving average processes that has been recently studied in \cite{BLP}. More specifically, we combine (negative) power variation statistics and empirical characteristic functions to obtain consistent estimates of $(\sigma, \alpha, H)$. We present the law of large numbers and some fully feasible weak limit theorems.

\ \

\bigskip


\end{abstract}

\section{Introduction} \label{sec1}
\setcounter{equation}{0}
\renewcommand{\theequation}{\thesection.\arabic{equation}}

Since the pioneering work by Mandelbrot and van Ness \cite{MV} fractional Brownian motion (fBm) became one of the most prominent Gaussian processes in the probabilistic and statistical literature. 
As a building block in stochastic models it found various applications in natural and social sciences 
such as physics, biology or economics. Mathematically speaking, the scaled fBm is fully characterised by its scaling parameter $\sigma>0$ and Hurst parameter $H \in (0,1)$. More specifically, the scaled 
fBm $Z_t = \sigma B_t^H$ is a zero mean Gaussian process with covariance kernel determined by 
\begin{align*}
\E\left[B_t^H B_s^H \right] = \frac 12 \left( t^{2H} + s^{2H} -|t-s|^{2H}  \right), \qquad t,s \geq 0.
\end{align*}  
We recall that the (scaled) fBm with Hurst parameter $H\in (0,1)$ is the unique Gaussian process with stationary increments and self-similarity index $H$, i.e. it holds that $(a^H Z_t) _{t \geq 0} 
= ( Z_{at}) _{t \geq 0}$ in distribution for any $a>0$. Over the last forty years there has been a lot of progress in limit theorems and statistical inference for fBm's. The estimation of the Hurst parameter $H$
and/or the scaling parameter $\sigma$ has been investigated in numerous papers both in low and high frequency framework. We refer to \cite{D} for efficient estimation of the Hurst parameter $H$ in the 
low frequency setting and to \cite{BF,CI,IL} for the estimation of $(\sigma, H)$ in the high frequency setting, among many others. In the low frequency framework the spectral density methods are usually applied and the optimal convergence rate for the estimation of $(\sigma, H)$ is known to be $\sqrt{n}$. 
In the high frequency setting the estimation of the pair $(\sigma, H)$ typically relies upon power variations and related statistics, and the optimal convergence rate is known to be $(\sqrt{n}/ \log(n),\sqrt{n})$.  More recently, the class of multifractional Brownian motions, which accounts for time varying Hurst parameter, has been introduced in the literature (see e.g. \cite{ACL,PL,ST}). We refer to the work 
\cite{BS,LP} for estimation techniques for the regularity of a  multifractional Brownian motion.  

If we drop the Gaussianity assumption the class of stationary increments self-similar processes becomes much larger. This is a consequence of the work by Pipiras and Taqqu  \cite{PT}, which in turn applies the decomposition results from the seminal paper by Rosi\'nski \cite{R} (see also \cite{S}).  
The crucial  theorem proved in  \cite{R} shows that each stationary stable process can be uniquely decomposed (in distribution) into three independent parts: the mixed moving average process, the harmonizable process and the ``third kind'' process described by a conservative nonsingular flow.  
The most prominent example of a non-Gaussian stationary increments self-similar process is the 
linear fractional stable motion (an element of the first class), which has been introduced in \cite{CM}. 
It is defined as follows: On a filtered probability space 
$(\Omega, \mathcal F, (\mathcal{F}_t)_{t \in \R}, \mathbb P)$, we introduce the process
\begin{align} \label{fLm}
X_t = \int_{\R} \left\{(t-s)_+^{H-1/\alpha} - (-s)_+^{H-1/\alpha} \right\} dL_s, \qquad x_+:= \max\{x,0\},
\end{align}
where $L$ is a symmetric $\alpha$-stable L\'evy motion, $\alpha \in (0,2)$, with scale parameter 
$\sigma>0$ and $H \in (0,1)$ (here we use the convention $x_+^a=0$ for any $x\leq 0$ and 
$a \in \R$). In some sense the linear fractional stable motion is a non-Gaussian analogue of fBm. The process $(X_t)_{t \in \R}$ has symmetric $\alpha$-stable marginals, stationary increments and it is self-similar with parameter $H$. Fractional stable motions are often used in natural sciences, e.g. in physics or internet traffic, where the process under consideration exhibits stationarity and self-similarity 
along with heavy tailed marginals (see e.g. \cite{GLT} for the context of turbulence modelling).
The probabilistic properties of linear fractional stable motions, such as  integration concepts, path and variational properties, have been intensively studied in several papers, see for example \cite{BCI,BLS,BM} among many others. However, from 
the statistical point of view, very little is known about the inference for the parameter 
$\theta=(\sigma, \alpha, H) \in \R_+ \times (0,2) \times (0,1)$ in high or low frequency setting. The few
existing papers mostly concentrate on estimation of the self-similarity parameter $H$. The work \cite{AH,PTA} investigates the asymptotic theory for a wavelet-based estimator of $H$ when 
$\alpha \in (1,2)$. In \cite{BLP, SPT} the authors suggest to use power variation statistics to obtain an estimator of $H$, but this method also requires the a priori knowledge of the lower bound for the stability parameter $\alpha$. Recently, the work \cite{DI} suggested to use negative power variations 
to get a consistent estimator of $H$, which applies for any $\alpha \in (0,2)$, but  this article does not contain a central limit theorem for this method. Finally, in \cite{BLP, GLT} the authors propose to use an empirical scale function to estimate the pair $(\alpha, H)$. However, this approach only provides 
a $\log(n)$-consistent estimator without any hope for a central limit theorem. 

In this paper we will propose a new estimation procedure for the parameter 
$\theta=(\sigma, \alpha, H) $ in high and low frequency framework. Our methodology is based upon the use of power variation statistics, with possibly negative powers, and  
the empirical characteristic function. The probabilistic techniques  originate from the recent article 
\cite{BLP}, which has developed the asymptotic theory for power variations 
of higher order differences of
stationary increments L\'evy moving averages (see also \cite{PT2, PTA} for related asymptotic theory). 
However, we will need to derive  much more complex asymptotic results to obtain a complete distributional theory 
for the estimator of the parameter $\theta \in \R_+ \times (0,2) \times (0,1)$. We will obtain a fully feasible asymptotic theory for our estimator with convergence rates $(\sqrt{n}, \sqrt{n},\sqrt{n})$ 
in the low frequency setting and $(\sqrt{n}/\log(n), \sqrt{n}/\log(n),\sqrt{n})$ in the high frequency setting.

The paper is structured as follows. Section \ref{sec2} presents the basic properties of the linear fractional stable motion, the review of the probabilistic results from \cite{BLP} and a multivariate limit theorem, which plays a key role for the statistical estimation. Section \ref{sec3} is devoted to the statistical inference in the continuous case $H-1/\alpha>0$. The general case is treated in Section \ref{sec4}. Finally, Section \ref{sec5} demonstrates some simulation results.   All proofs are collected
in Section \ref{sec6}.

\section{First properties and some asymptotic results} \label{sec2}
\setcounter{equation}{0}
\renewcommand{\theequation}{\thesection.\arabic{equation}}

\subsection{Distributional and path properties} \label{sec2.1}
In this section we review some basic properties of the linear  fractional  stable motion. First of all, we
recall that the symmetric $\alpha$-stable process $(L_t)_{t \in \R}$ with scale parameter $\sigma>0$ is uniquely determined by the characteristic function of $L_1$, which is given by
\begin{align} \label{charL}
\E[\exp(itL_1)] = \exp(-\sigma^{\alpha} |t|^{\alpha}), \qquad t \in \R.
\end{align}
Following the theory of integration with respect to infinitely divisible processes investigated in \cite{RR},
we know that for any deterministic function $g : \R \to \R$
\begin{align*}
X= \int_{\R} g_s dL_s < \infty \quad \text{almost surely} \qquad 
\Leftrightarrow \qquad \| g \|_{\alpha}^{\alpha}:=\int_{\R} |g_s|^{\alpha} ds < \infty.
\end{align*}
Furthermore, if $\| g \|_{\alpha}< \infty$ then $X$ has a symmetric $\alpha$-stable distribution 
with scale parameter $\sigma\| g \|_{\alpha}$. In particular, setting 
\begin{align} \label{rep}
X_t = \int_{\R} g_t(s) dL_s, \qquad g_t(s):=\left\{(t-s)_+^{H-1/\alpha} - (-s)_+^{H-1/\alpha} \right\},
\end{align}
we see that $g_t \in L^{\alpha}(\R)$ for any $t \in \R$, since $|g_t(s)| \leq C_t |s|^{H-1-1/\alpha}$
when $s \to - \infty$ and $H\in (0,1)$. Hence, $X_t$ is well defined for any $t \in \R$ and all finite 
dimensional distributions of the linear  fractional  stable motion $(X_t)_{t \in \R}$ are symmetric 
$\alpha$-stable. It is easily seen that the linear fractional stable motion has stationary increments.

We recall that symmetric $\alpha$-stable random variables with $\alpha \in (0,2)$ do not exhibit finite
second moments, and hence their dependence structure can't be measured via the classical covariance kernel. Instead it is often useful to consider the following measure of dependence. Let 
$X= \int_{\R} g_s dL_s$ and $Y= \int_{\R} h_s dL_s$ with $\| g \|_{\alpha}, \| h\|_{\alpha}< \infty$. Then 
we introduce the measure of dependence $U_{g,h}: \R^2 \to \R$ via
\begin{align} \label{U}
U_{g,h} (u,v)&:= \E[\exp(i(uX+vY))] - \E[\exp(iuX)] \E[\exp(ivY)] \\[1.5 ex]
&= \exp(- \sigma^{\alpha} 
\| ug +vh \|_{\alpha}^{\alpha} ) - \exp(- \sigma^{\alpha} (\| ug \|_{\alpha}^{\alpha} + \| vh \|_{\alpha}^{\alpha})).\nonumber
\end{align}
The quantity $U_{g,h}$ is extremely useful when computing covariances $\text{cov}(K_1(X), K_2(Y))$
for functions $K_1, K_2 \in L^1(\R)$;
see for instance \cite{PTA}. Let $\mathfrak{F}$ denote the Fourier transform and let  $\mathfrak{F}^{-1}$ be its inverse. Furthermore, let $p_{(X,Y)}$, $p_{X}$ and $p_{Y}$ denote the density of $(X,Y)$,
$X$ and $Y$, respectively. We recall that these densities are not available in a closed form except in some special cases. 
Using the duality relationship we obtain the identity 
\begin{align}
\text{cov}(K_1(X), K_2(Y)) &= \int_{\R^2} K_1(x) K_2(y) \left(p_{(X,Y)} (x,y) - p_X(x) p_Y(y) \right) dxdy
\nonumber \\[1.5 ex]
\label{covFourier} &= \int_{\R^2} K_1(x) K_2(y) \mathfrak{F}^{-1} U_{g,h} (x,y) dxdy \\[1.5 ex]
&=  \int_{\R^2} \left(\mathfrak{F}^{-1}  K_1(x) \right) \left(\mathfrak{F}^{-1}  K_2(y) \right)  U_{g,h} (x,y) dxdy. \nonumber
\end{align}
We remark that the latter provides an explicit formula for computation of covariances  
$\text{cov}(K_1(X), K_2(Y))$. 

Finally, we recall that the path properties of a linear fractional stable motion strongly depend on the interplay between the parameters $H$ and $\alpha$. When $H-1/\alpha>0$ the process 
$(X_t)_{t \in \R}$ is H\"older continuous on compact intervals of any order smaller than $H-1/\alpha$;
we refer to \cite{BCI} for more details on this property. If  $H-1/\alpha<0$ the 
linear fractional stable motion explodes at jump times of the driving L\'evy process $L$; in particular, 
$X$ has unbounded paths on compact intervals. 
We demonstrate some sample paths of the linear fractional stable motions in Figure \ref{fig1}. 
In the critical case $H-1/\alpha=0$ we obviously have
the identity $X_t=L_t$. In this situation the parameter estimation  has been investigated in  \cite{AJ}.

\begin{figure}
\centering
                \includegraphics[width=1\textwidth]{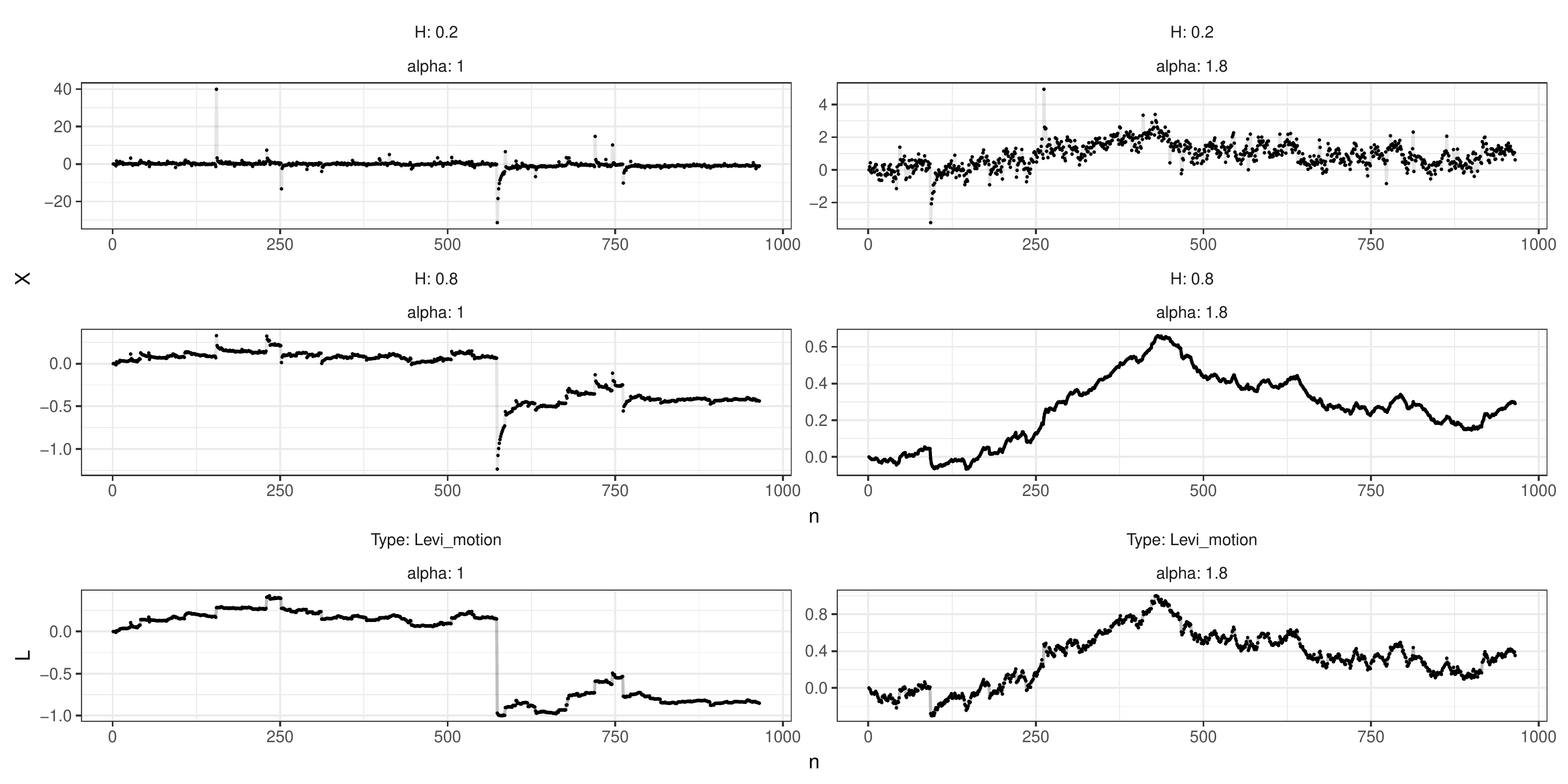}
\caption{\footnotesize Left (from bottom to top): The driving symmetric stable L\'evy process with $\alpha=1$, 
linear fractional stable motions with parameters $\alpha=1, H=0.8$ and $\alpha=1, H=0.2$.
Right (from bottom to top): The driving symmetric stable L\'evy process with $\alpha=1.8$, 
linear fractional stable motions with parameters $\alpha=1.8, H=0.8$ and $\alpha=1.8, H=0.2$.}
\label{fig1}
\end{figure}

\subsection{Review of the limit theory} \label{sec2.2}
In this section we review some probabilistic results, which will be relevant for our estimation method. Due to stationarity of the increments and self-similarity of the process $(X_t)_{t \in \R}$, we can
discuss the limit theory for the high and low frequency case simultaneously.  We start by introducing higher order increments of $X$.  We denote by $\Delta_{i,k}^{n,r} X$ ($i,k,r,n \in \N$)  the  $k$th order increment of $X$ at stage $i/n$ and frequency $r/n$, i.e. 
\begin{align} \label{filterhigh}
\Delta_{i,k}^{n,r} X:= \sum_{j=0}^k (-1)^j \binom{k}{j} X_{(i-rj)/n}, \qquad i\geq rk.
\end{align} 
Note that for $r=k=1$ we obtain the usual increments 
$\Delta_{i,1}^{n,1} X = X_{i/n}-X_{(i-1)/n}$. For the ease of notation we will often drop the index
$r$ (resp. $k$ and $n$) in $\Delta_{i,k}^{n,r} X$ and other quantities when $r=1$ (resp. $k=1$ and $n=1$). In particular, 
the low frequency $k$th order increments of $X$ are denoted by
\begin{align} \label{filterlow}
\Delta_{i,k}^{r} X:= \sum_{j=0}^k (-1)^j \binom{k}{j} X_{i-rj}, \qquad i\geq rk.
\end{align} 
According to the self-similarity of the process $(X_t)_{t \in \R}$ we readily have that 
$(n^H \Delta_{i,k}^{n,r} X)_{i \geq rk} \eqschw ( \Delta_{i,k}^{r} X)_{i \geq rk}$. Our main probabilistic 
tools will be statistics of the form
\begin{align}  \label{Vstat}
V_{\text{high}}(f; k,r)_n := \frac{1}{n} \sum_{i=rk}^n f\left( n^H \Delta_{i,k}^{n,r} X \right), 
\qquad V_{\text{low}}(f; k,r)_n := \frac{1}{n} \sum_{i=rk}^n f\left( \Delta_{i,k}^{r} X \right),
\end{align}
where $f: \R \to \R$ is a measurable function. It is well known that the process $(X_t)_{t \in \R}$ 
is mixing, see e.g. \cite{CHW}. Hence, Birkhoff's ergodic theorem implies the convergence
$V_{\text{low}}(f; k,r)_n \to \E[f(\Delta_{rk,k}^{r} X)]$ almost surely whenever 
$\E[|f(\Delta_{rk,k}^{r} X)|]< \infty$. The same result holds in probability for the statistic $V_{\text{high}}(f; k,r)_n$ due to self-similarity of the process $X$. 
However, the weak limit theorems associated with the aforementioned law of large numbers and 
the framework of functions $f$ with $\E[|f(\Delta_{rk,k}^{r} X)|]=\infty$ are not completely understood 
in the literature. To get an idea about possible limits that may appear we briefly demonstrate some recent theoretical developments from the paper \cite{BLP}, where the case $f_p(x)=|x|^p$ ($p>0$) has
been investigated. We remark that their results are obtained for a wider class of processes, namely
stationary increments L\'evy moving average processes, and we adapt them to the setting of linear
fractional stable motions. 

We need to introduce some more notation to describe the various limits. 
For $p \in (-1,1) \setminus \{0\}$ we define the constant
\begin{align} \label{ap}
a_p:= 
\begin{cases}
 \int_{\R} \left(1-\cos(y) \right) |y|^{-1-p} dy: & p \in (0,1) \\[1.5 ex]
 \sqrt{2 \pi} \Gamma(-p/2)/ 2^{p+1/2} \Gamma((p+1)/2): & p \in (-1,0)
\end{cases}
~,
\end{align}
where $\Gamma$ denotes the Gamma function. It is easy to see that $a_p>0$ is indeed finite in all relevant cases. 
For
any functions $g,h\in L^{\alpha}(\R)$, we introduce the notation 
\begin{align} \label{powercov}
\theta(g,h)_{p} = a_p^{-2} \int_{\R^2}  |xy|^{-1-p}U_{g,h}(x,y) dxdy,
\end{align} 
where $U_{g,h}$ is defined in \eqref{U}, whenever the above double integral is finite. 
 Furthermore,  for $k,r \in \N$, we define the function $h_{k,r}:\R \to \R$ by
\begin{align} \label{def-h}
h_{k,r}(x)=  \sum_{j=0}^k (-1)^j \binom{k}{j} (x-rj)_{+}^{H-1/\alpha},\qquad x\in \R.
\end{align}  
Below $(U_m)_{m\geq 1}$ is an i.i.d. $\mathcal{U}(0,1)$-distributed sequence of random variables
independent of $L$, $(T_m)_{m \geq 1}$ are jump times of $L$ and $\Delta L_{T_m}:= 
L_{T_m} -L_{T_m-}$ are jump sizes. The following result summarises the limit theory for the statistic
$V_{\text{high}}(f_p; k)_n$ (i.e. $r=1$) in the power variation setting. 

\begin{theo} \label{th1}
(\cite[Theorems 1.1 and 1.2]{BLP}) We consider the function $f_p(x)=|x|^p$ ($p>0$) and assume that
$H-1/\alpha>0$.  \\
(i) (First order asymptotics) If  $p>\alpha$ we obtain convergence in law 
\[
n^{1 -p/\alpha} V_{\text{high}}(f_p; k)_n \schw 
\sum_{m:\, T_m\in [0,1]} |\Delta L_{T_m}|^p \left( \sum_{l=0}^{\infty} |h_{k}(l+U_m)|^p \right).
\]
If $p<\alpha$ we deduce the  law of large numbers
\[
V_{\text{high}}(f_p; k)_n \toop m_{p,k}:= \E[|\Delta_{k,k} X|^p]. 
\]
(ii) (Second order asymptotics) Assume that $p< \alpha/2$. If $H< k-1/\alpha$ we obtain 
the central limit theorem
\[
\sqrt{n} \left( V_{\text{high}}(f_p; k)_n - m_{p,k} \right) \schw \mathcal N(0, \eta^2), \qquad
 \eta^2 = \theta(h_{k}, h_{k})_p +2 \sum_{j=1}^{\infty} 
\theta(h_{k}, h_{k} (\cdot+j))_p,
\]
where the quantity $\theta(g,h)$ has been introduced at \eqref{powercov}.
If $H >k-1/\alpha$ we deduce a non-central limit theorem 
\[
n^{1-1/(1+ \alpha(k-H))}\left(V_{\text{high}}(f_p; k)_n - m_{p,k} \right) \schw S,
\]
where $S$ is a totally right skewed  $(1+ \alpha(k-H))$-stable random variable with mean zero and scale parameter  $\widetilde{\sigma}$, which is defined in \cite[Theorem 1.2]{BLP}. 
\end{theo}

We remark that the results of Theorem \ref{th1} remain valid for the low frequency statistic 
$V_{\text{low}}(f_p; k)_n$ due to self-similarity property of $L$. Apart from various critical cases Theorem \ref{th1} gives a rather complete understanding of the asymptotic behaviour of the power
variation $V_{\text{high}}(f_p; k)_n$ in the setting $H-1/\alpha>0$. The strong law of large numbers
in Theorem \ref{th1}(i) will be useful for estimation of the parameter $H$. However, without an a priori 
knowledge about the stability parameter $\alpha$, we can't insure that the condition $p<\alpha$ holds.
Similarly, we would like to use the central limit theorem in  Theorem \ref{th1}(ii) whose convergence
rate $\sqrt{n}$ is faster than the rate $n^{1-1/(1+ \alpha(k-H))}$ in the non-central limit theorem. 
But the conditions of  Theorem \ref{th1}(ii) rely again on an a priori knowledge about $\alpha$. 

There are some few related results in the literature. In \cite{PT2} the authors have shown a central limit
theorem a standardised version of the statistic $\sum_{i=1}^n f(Y_i)$, where $f$ is a \textit{bounded}
function and $(Y_t)_{t \in \R}$ is a stable moving average process. In a later work \cite{PTA} the result
has been extended to a certain class of unbounded functions $f$ under the additional assumption 
that $\alpha \in (1,2)$. Similarly to Theorem \ref{th1} the sufficient conditions for the validity of the 
central limit theorems in \cite{PT2,PTA} depend on the interplay between the kernel function of the stable moving average process and the stability index $\alpha$.   We remark that extensions of these results in various directions will be necessary to obtain the full asymptotic theory for estimators of the
parameter  $\theta=(\sigma, \alpha, H) $.

\subsection{A multivariate weak limit theorem}

Although Theorem \ref{th1}(ii) gives a rather complete picture of the weak limit theory in the power variation case, we will require a much stronger result for our statistical applications. We introduce the function $\psi_t : \R \to \R$ with $\psi_t(x)= \cos(tx)$ and define the statistics 
\begin{align} \label{ecf}
\varphi_{\text{high}}(t; H,k)_n := V_{\text{high}}(\psi_t; k)_n \qquad \text{and}
\qquad \varphi_{\text{low}}(t; k)_n := V_{\text{low}}(\psi_t; k)_n,
\end{align}
which correspond to $r=1$. Notice that, in contrast to $\varphi_{\text{low}}(t; k)_n$, the high frequency statistic $\varphi_{\text{high}}(t; H,k)_n$ depends on the unknown self-similarity parameter $H$. In fact, this is the major difference between the high and low frequency settings, which will result in different rates of convergence later on. Applying again the strong law 
of large numbers  we readily obtain the strong consistency 
\begin{align} \label{ecfconv}
\varphi_{\text{low}}(t; k)_n \toas \varphi(t; k):= \exp \left( -|\sigma \|h_{k}\|_{\alpha} t|^ \alpha \right).
\end{align}
Clearly,  the same result holds in probability for the high frequency statistic 
$\varphi_{\text{high}}(t; H,k)_n$. Next, we introduce various types of statistics, which will play a major role in estimation of the unknown parameter $\theta$. More specifically, we will extend the definition of power variation to certain negative powers  and prove a multivariate limit theorem for power variations and empirical characteristic functions. 
We fix $d \in \N$ and define the statistics for any $1 \leq j \leq d$, $r_j \in \{1,2\}$, $p \in (-1/2, 1/2) \setminus \{0\}$ and $t_j>0$:
\begin{align} \label{multstat}
&\left.
\begin{array}{c}
W(n)^{(1)}_j := \sqrt{n} \left( V_{\text{low}} (f_{p}; k_j,r_j) - r_j^H m_{p,k_j}\right) \\[1.5 ex]
W(n)^{(2)}_j := \sqrt{n} \left( V_{\text{low}}(\psi_{t_j}; k_j)_n - 
\varphi(t_j; k_j) \right) 
\end{array}
\right \} \qquad \text{when } k_j> H+ 1/\alpha \\[1.5 ex]
&\left.
\begin{array}{c}
S(n)^{(1)}_j := n^{1-1/(1+ \alpha(k-H))}\left( V_{\text{low}} (f_{p}; k,r_j) - r_j^H m_{p,k}\right) \\[1.5 ex]
S(n)^{(2)}_j := n^{1-1/(1+ \alpha(k-H))} \left( V_{\text{low}}(\psi_{t_j}; k)_n - 
\varphi(t_j; k) \right)
\end{array}
\right \} \qquad \text{when } k< H+ 1/\alpha \nonumber 
\end{align} 
Note the identity $\E[|\Delta_{rk,k}^r X|^p]=r^{H}m_{p,k}$, which explains the centring of the statistics  
$W(n)^{(1)}$ and $S(n)^{(1)}$.
We remark that the functionals $W(n)^{(1)}$ and $W(n)^{(2)}$ are in the domain of attraction of the normal distribution (under appropriate assumption on the powers $p$) while the functionals $S(n)^{(1)}$ and $S(n)^{(2)}$ are in the domain of attraction of the $(1+\alpha(k-H))$-stable distribution. The latter fact is rather surprising since the statistic $S(n)^{(2)}_j$ exhibits finite moments of any order. 

Before we proceed with the main result of this section we need to introduce some more notation. In the first step, for any $x \in \R$, we define the functions
\begin{align} \label{Phi12}
\Phi_{j}^{(1)}(x) &= \E[f_p(\Delta_{r_jk,k}^{r_j} X + x)] - \E[f_p(\Delta_{r_jk,k}^{r_j} X )], \\[1.5 ex]
\Phi_{j}^{(2)}(x) &= \E[\psi_{t_j}(\Delta_{k,k} X + x)] - \E[\psi_{t_j}(\Delta_{k,k} X )]. \nonumber
\end{align}
Since the functions $f_p$ and $\psi_t$ are even we readily obtain that $\Phi_{j}^{(l)}(0)= 
\nabla \Phi_{j}^{(l)}(0)=0$ for all $l,j$. Thus, using Lemma \ref{lemPhi}, we deduce the growth estimates
\begin{align} \label{growth}
|\Phi_{j}^{(1)}(x)| \leq C\left(x^2 \wedge  |x|^{\max\{p,0\}}  \right), \quad 
|\Phi_{j}^{(2)}(x)| \leq C\left(x^2 \wedge  1\right),
\end{align} 
for some positive constant $C$. 
Next, we introduce the functions
\begin{align} \label{barPhi12}
\overline{\Phi}_j^{(1)} (x) = \sum_{i=1}^{\infty} \Phi_{j}^{(1)}\left(h_{k,r_j}(i)x\right), \qquad 
\overline{\Phi}_j^{(2)} (x) = \sum_{i=1}^{\infty} \Phi_{j}^{(2)}\left(h_{k}(i)x \right).
\end{align} 
Note that these functions are indeed finite due to \eqref{growth} and
the estimate $|h_{k,r}(x)|\leq C |x|^{H-1/\alpha-k}$ for large $x$. Finally,
we set $\overline{\Phi}= (\overline{\Phi}^{(1)}, \overline{\Phi}^{(2)}) =(\Phi_{1}^{(1)}, \ldots, \Phi_{d}^{(1)}, \Phi_{1}^{(2)}, \ldots, \Phi_{d}^{(2)})$.
The main probabilistic result of this paper is the following theorem.

\begin{theo} \label{multivariateclt}
Assume that either $p \in (-1/2,0)$ or $p \in (0,1/2)$ and $p< \alpha/2$. Set
$W(n)^{(i)} =(W(n)^{(i)}_1 , \ldots, W(n)^{(i)}_d)$ and 
$S(n)^{(i)} =(S(n)^{(i)}_1 , \ldots, S(n)^{(i)}_d)$ for $i=1,2$. Then we obtain weak convergence in law
on $\R^{4d}$:
\begin{align} \label{mclt}
\left(W(n)^{(1)}, W(n)^{(2)}, S(n)^{(1)}, S(n)^{(2)} \right) \schw 
\left(W^{(1)}, W^{(2)}, S^{(1)}, S^{(2)} \right),  
\end{align}
where $W=(W^{(1)}, W^{(2)})$ and $S=(S^{(1)}, S^{(2)})$ are independent, $W$ is a centred $2d$-dimensional normal distribution with covariance matrix determined by
\begin{align*}
\text{\rm cov}\left(W^{(i)}_j, W^{(i')}_{j'} \right) = \lim_{n \to \infty} 
\text{\rm cov}\left(W(n)^{(i)}_j, W(n)^{(i')}_{j'} \right) \qquad 1 \leq j,j'\leq d, ~ i,i'=1,2, 
\end{align*}
and $S^{(1)}$, $S^{(2)}$ are independent $d$-dimensional $(1+\alpha(k-H))$-stable random variables.
The law of $S^{(1)}$ (resp. $S^{(2)}$) is determined by the L\'evy measure $\nu_1$ (resp. $\nu_2$) whose support is the cone $(\R_+)^{d}$ (resp. $(\R_-)^{d}$). More specifically, for any Borel sets $A_1 \in
(\R_+)^{d}$, $A_2 \in
(\R_-)^{d}$ bounded away from $0$ the quantities $\nu_1(A_1), \nu_2(A_2)$ are determined by the identity
\begin{align} \label{nu}
\nu_l (A_l) = \lim_{n \to \infty} n \mathbb P\left(n^{-1/(1+ \alpha(k-H))} \overline{\Phi}^{(l)} (L_1) \in A_l \right) , \qquad l=1,2.
\end{align}
\end{theo}
The probabilistic result of Theorem \ref{multivariateclt} is new in the literature; neither the negative power variations nor the (real part of) empirical characteristic function have been studied from the distributional perspective. We remark that the statistics $W(n)^{(1)}$ and $S(n)^{(1)}$ use the same  powers $p$ while 
the quantities  $S(n)^{(1)}$ and $S(n)^{(2)}$ are based on the same order of increments $k$. The result of Theorem \ref{multivariateclt} does not really use these particular restrictions, but its statement is sufficient for the statistical application under investigation. 

There exists an explicit expression for the covariance matrix of the limit $W$. We obtain the following representations:
\begin{align} \label{asycov}
\text{\rm cov}\left(W^{(1)}_j, W^{(1)}_{j'} \right) &= \sum_{l \in \Z} 
 \theta(h_{k_{j},r_j}, h_{{k_{j'}},r_{j'}}(\cdot+l))_{p}, \\
\text{\rm cov}\left(W^{(2)}_j, W^{(2)}_{j'} \right)&= \frac 12 \sum_{l \in \Z} \left(U_{h_{k_j},h_{k_{j'}}(\cdot+l)} (t_{j},t_{j'})
+ U_{h_{k_j},-h_{k_{j'}}(\cdot+l)} (t_{j},t_{j'}) \right) ,  \nonumber \\
\text{\rm cov}\left(W^{(1)}_j, W^{(2)}_{j'} \right) &= \sum_{l \in \Z} \overline{\theta} (l)_{jj'},  \nonumber
\end{align}
with 
\begin{align*}
\overline{\theta} (l)_{jj'}= -a_{p}^{-1} \int_{\R} |y|^{-1-p} 
U_{h_{k_j,r_j},h_{k_{j'}}(\cdot+l)} (y,t_{j'})
dy.
\end{align*}
We will prove that $\text{\rm cov}(W)<\infty$ in all relevant cases 
 and the mapping $(\sigma, \alpha, H) \mapsto \text{\rm cov}(W)$ is continuous (see Section \ref{sec6.1}). In principle, the latter allows us to estimate  the covariance matrix $\text{\rm cov}(W)<\infty$ and thus to obtain a feasible version 
of the central limit theorem in Theorem \ref{multivariateclt}, although we will use a different approach in the simulation study. 

Similarly, the L\'evy measures $\nu_l$  ($l=1,2$)  can be determined explicitly. First of all, the representation \eqref{xp} from Section \ref{sec6.1} implies the identities
\begin{align*} 
\Phi_{j}^{(1)}(x) &= a_p^{-1} \int_{\R} \left(1- \cos(ux) \right)
 \exp\left( -|\sigma  \|h_{k,r_j}\|_{\alpha} u|^{\alpha}\right) |u|^{-1-p} du, \\[1.5 ex]
\Phi_{j}^{(2)}(x) &= (\cos(t_j x) - 1) \exp\left( -|\sigma  \|h_{k}\|_{\alpha} t_j|^{\alpha} \right). 
\end{align*}
In particular, it holds that $\Phi_{j}^{(1)}(x) \geq 0$ and $\Phi_{j}^{(2)}(x) \leq 0$. In the next step we need to determine the asymptotic behaviour of $\overline{\Phi}_j^{(1)} (x)$  
(resp. $\overline{\Phi}_j^{(2)} (x)$) as $x \to \infty$ (resp. as $x \to -\infty$). By the substitution 
$u=(x/z)^{1/(k+1/\alpha - H)}$ we have that 
\begin{align} \label{cj1}
&x^{1/(H-k-1/\alpha)} \overline{\Phi}_j^{(1)} (x)=  x^{1/(H-k-1/\alpha)}   \int_0^\infty 
\Phi_j^{(1)} \left(h_{k,r_j} (\lfloor u \rfloor +1)x \right) du \nonumber \\[1.5 ex] 
&=   (k+1/\alpha - H)^{-1}\int_0^\infty 
\Phi_j^{(1)} \left(h_{k,r_j} (\lfloor (x/z)^{1/(k+1/\alpha - H)} \rfloor +1)x \right) 
z^{-1+1/(H-k-1/\alpha)}dz  \nonumber \\[1.5 ex] 
&\to c_{j}^{(1)}:=  (k+1/\alpha - H)^{-1} \int_0^\infty 
\Phi_j^{(1)} \left(r_j^k \prod_{i=0}^{k-1} (H-1/\alpha -i) \cdot z \right) z^{-1+1/(H-k-1/\alpha)}dz
\end{align}
as $x \to \infty$. The convergence at \eqref{cj1} follows from the asymptotic behaviour $h_{k,r_j}(x)
\sim r_j^k \prod_{i=0}^{k-1} (H-1/\alpha -i) \cdot x^{H-1/\alpha -k}$ as $x \to \infty$. Applying the same technique we deduce that 
\begin{align} \label{cj2}
&|x|^{1/(H-k-1/\alpha)} \overline{\Phi}_j^{(2)} (x) \to \nonumber \\[1.5 ex]
& c_{j}^{(2)}:=  (k+1/\alpha - H)^{-1} \int_0^\infty 
\Phi_j^{(2)} \left( \prod_{i=0}^{k-1} (H-1/\alpha -i) \cdot z \right) z^{-1+1/(H-k-1/\alpha)}dz
\end{align}
as $x \to - \infty$. Now, both measures $\nu_1$ and $\nu_2$ from Theorem \ref{multivariateclt} can be related to the L\'evy measure $\nu$ of $L$. We introduce the mappings $\tau_1: \R_+ \to (\R_+)^d$
and $\tau_2: \R_- \to (\R_-)^d$ via
\begin{align*}
\tau_1 (x) = x^{1/(k+1/\alpha - H)} \left(c_1^{(1)}, \ldots, c_d^{(1)} \right), \qquad
\tau_2 (x) = |x|^{1/(k+1/\alpha - H)} \left(c_1^{(2)}, \ldots, c_d^{(2)} \right).
\end{align*}
Then, for Borel sets $A_1, A_2$ as defined in Theorem \ref{multivariateclt},  we deduce the identity
\begin{align} \label{nul}
\nu_l(A_l) = \lim_{n \to \infty} n \mathbb P\left(\tau_l(n^{-1/ \alpha} L_1) \in A_l \right) 
= \nu\left(\tau_l^{-1} (A_l) \right), \qquad l=1,2.
\end{align}

\section{Statistical inference in the continuous case $H-1/\alpha>0$} \label{sec3}
\setcounter{equation}{0}
\renewcommand{\theequation}{\thesection.\arabic{equation}}

We start with the continuous case $H-1/\alpha>0$, which turns out to be somewhat easier to treat compared to the general setting. Since $H \in (0,1)$ and $\alpha \in (0,2)$, condition $H-1/\alpha>0$ 
implies the restrictions
\[
\alpha \in (1,2) \qquad \text{and} \qquad H \in (1/2,1).
\]
It is the lower bound $\alpha>1$ that enables us to use the  law of large numbers in Theorem 
\ref{th1}(i) whenever $p<1$, and the central limit theorem in Theorem 
\ref{th1}(ii) whenever $p<1/2$ and $H<k-1/\alpha$. The latter condition $H<k-1/\alpha$
never holds for $k=1$ since $0<H-1/\alpha < 1-2/\alpha <0$ gives a contradiction, 
but it is  always satisfied for any $k \geq 2$ since
\[
H<1< k - 1/\alpha \qquad \text{for any } k \geq 2,
\] 
because $\alpha>1$.

Now, we introduce an estimator for the parameter $\theta=(\sigma, \alpha, H) $ in high and low frequency setting. We start with the statistical inference for the self-similarity parameter $H$, which is 
based upon a ratio statistic that compares power variations at two different frequencies. More specifically, we define the quantities 
\begin{align} \label{ratio}
R_{\text{high}} (p,k)_n := \frac{\sum_{i=2k}^n \left| \Delta_{i,k}^{n,2} X \right|^p}
{\sum_{i=k}^n \left| \Delta_{i,k}^{n,1} X \right|^p}, \qquad 
R_{\text{low}} (p,k)_n := \frac{\sum_{i=2k}^n \left| \Delta_{i,k}^{2} X \right|^p}
{\sum_{i=k}^n \left| \Delta_{i,k}^{1} X \right|^p},
\end{align}
where the increments $\Delta_{i,k}^{r} X$ have been defined at \eqref{filterlow}.
We obtain the convergence 
\[
R_{\text{high}} (p,k)_n \toop 2^{pH}, \qquad R_{\text{low}} (p,k)_n \toas 2^{pH}
\]
for any $p \in (0,1)$ as an immediate consequence of Theorem \ref{th1}(i). Consequently, defining 
the statistics
\begin{align} \label{Hest}
\widehat{H}_{\text{high}} (p,k)_n:= \frac 1p \log_2  R_{\text{high}} (p,k)_n, 
\qquad \widehat{H}_{\text{low}} (p,k)_n:= \frac 1p \log_2  R_{\text{low}} (p,k)_n,
\end{align} 
we deduce the consistency $\widehat{H}_{\text{high}} (p,k)_n \toop H$, 
$\widehat{H}_{\text{low}} (p,k)_n \toas H$ as $n \to \infty$ for any $k \geq 1$ and any $p \in (0,1)$. 
We remark that this type of ratio statistics is commonly used in the framework of fBm's when 
estimating 
the Hurst parameter $H$ (see e.g. \cite{IL} among many others). In the Gaussian setting, which corresponds to $\alpha=2$, the central 
limit theorem for the quantity $\sqrt{n}(\widehat{H}_{\text{high}} (p,k)_n - H)$  holds for all $k \geq 2$
and also for $k=1$ if further $H\in (0, 3/4)$. As we indicated above, in the framework of pure jump
$\alpha$-stable driving motion $L$ the central limit theorem never holds if $k=1$. Hence, there is no
smooth transition between the non-Gaussian and Gaussian setting when $\alpha \to 2$.

The estimation strategy for the parameter $\theta=(\sigma, \alpha, H)$ based on high frequency observations is now straightforward: Infer the self-similarity parameter $H$ by  \eqref{Hest} and 
use the plug-in estimator $\varphi_{\text{high}}(t; \widehat{H}_{\text{high}}(p,k),k)_n$ for two different
values of $t$ to infer the scale parameter $\sigma$ and the stability index $\alpha$. For the latter step
we consider $t_2>t_1>0$ and observe the identities 
\begin{align*}
 \sigma= \left(- \log  \varphi(t_1; k) \right)^{1/\alpha} /  t_1 \|h_{k}\|_{\alpha} , \quad
\alpha = \frac{\log |\log \varphi(t_2; k)| - \log |\log \varphi(t_1; k)| } {\log t_2 - \log t_1}.
\end{align*}
Recalling that the function $h_k$ depends on $\alpha$ and $H$, we readily  obtain a function $G$ such that 
\begin{align} \label{G}
(\sigma, \alpha ) = G\left(\varphi(t_1; k), \varphi(t_2; k), H \right)
\end{align} 
where we applied the above identities. Next, we present the estimator of the pair $(\sigma, \alpha)$ in high and low frequency setting, recalling 
that the estimators of the self-similarity parameter $H$ have been defined at   \eqref{Hest}.  
We introduce the following estimators:
\begin{align} 
&\left( \widehat{\sigma}_{\text{high}} (k,t_1,t_2)_n, \widehat{\alpha}_{\text{high}} (k,t_1,t_2)_n \right)
 \nonumber \\[1.5 ex]
&= G\left( \varphi_{\text{high}}(t_1; \widehat{H}_{\text{high}} (p,k)_n,k)_n, 
\varphi_{\text{high}}(t_2; \widehat{H}_{\text{high}} (p,k)_n,k)_n, 
\widehat{H}_{\text{high}} (p,k)_n \right), \nonumber \\[1.5 ex]
\label{alphasigma}
&\left( \widehat{\sigma}_{\text{low}} (k,t_1,t_2)_n, \widehat{\alpha}_{\text{low}} (k,t_1,t_2)_n \right)
= G\left( \varphi_{\text{low}}(t_1; k)_n, 
\varphi_{\text{low}}(t_2; k)_n, \widehat{H}_{\text{low}} (p,k)_n \right).
\end{align}
Before we present the main result of this section we need to introduce more notation. We define 
the functions $v_p : \R_{+}^2 \to \R$ and $F: \R_{+}^2 \times \R^2 \to \R^3$ by 
\begin{align} \label{vF}
v_p(x,y)= p^{-1} (\log_2 y - \log_2 x), \qquad F(x,y,u,w) = \left(G(u,w, v_p(x,y)), v_p(x,y) \right),
\end{align}
and let $JF$ denotes the Jacobian of $F$. For any matrix $A$ we write  $A^{\star}$ for its transpose.
The asymptotic normality in the low and high frequency setting is summarised in the following theorem.

\begin{theo} \label{th3} Consider the linear fractional stable motion $(X_t)_{t \in \R}$
introduced at \eqref{fLm}. Let $k \geq 2$ and $t_2>t_1>0$. 
\newline
(i) (\textit{Low frequency case}) Let $W=(W^{(1)}, W^{(2)})$ be the $4$-dimensional normal limit defined in Theorem \ref{multivariateclt}
associated with $d=2$, $p \in (0,1/2)$, $k_1=k_2=k$ and $r_j=j$. Then we obtain the central limit theorem
\[
\sqrt{n} \left( 
\begin{array} {c}
\widehat{\sigma}_{\text{\rm low}} (k,t_1,t_2)_n - \sigma \\
\widehat{\alpha}_{\text{\rm low}} (k,t_1,t_2)_n - \alpha \\
\widehat{H}_{\text{\rm low}} (p,k)_n - H 
\end{array}
\right)
\schw B_{\text{\rm low}}^{\text{\rm nor}}(p,k)=JF \left(m_{p,k}, 2^H m_{p,k}, \varphi(t_1;k), \varphi(t_2;k) \right) W^{\star}.
\]
(ii)  (\textit{High frequency case}) We obtain the central limit theorem
\begin{align*}
&\left( 
\begin{array} {c}
\sqrt{n} (\log n)^{-1} \left(\widehat{\sigma}_{\text{\rm high}} (k,t_1,t_2)_n - \sigma \right) \\
\sqrt{n} (\log n)^{-1} \left( \widehat{\alpha}_{\text{\rm high}} (k,t_1,t_2)_n - \alpha \right) \\
\sqrt{n}  \left(\widehat{H}_{\text{\rm high}} (p,k)_n - H \right) 
\end{array}
\right)
\schw B_{\text{\rm high}}^{\text{\rm nor}}(p,k)=\\
& \nabla v(m_{p,k}, 2^H m_{p,k}) (W^{(1)})^{\star}
\times \left(
\begin{array} {c}
\nabla G_1(\varphi(t_1;k), \varphi(t_2;k),H) \left(t_1\varphi'(t_1;k), t_2\varphi'(t_2;k),0 \right)^{\star} \\
\nabla G_2(\varphi(t_1;k), \varphi(t_2;k),H) \left(t_1\varphi'(t_1;k), t_2\varphi'(t_2;k),0 \right)^{\star} \\
1
\end{array} 
\right).
\end{align*}  
\end{theo}
We remark that the central limit theorem of Theorem \ref{th3}(i) is a simple consequence of 
Theorem \ref{multivariateclt} and the delta method. In contrast to the low frequency case Theorem \ref{th3}(ii) is degenerate 
in the sense that the limit distribution is solely driven by the asymptotics of the term 
$\sqrt{n}  (\widehat{H}_{\text{high}} (p,k)_n - H ) $. Since the parameter $H$ enters the quantity 
$\varphi_{\text{high}}(t; H,k)_n$ via $n^H$ the additional term $(\log n)^{-1}$ appears in the convergence rate. 

For a later use we need to extend the definition of the random variables $B_{\text{\rm high}}^{\text{\rm nor}}(p,k)$
and $B_{\text{\rm low}}^{\text{\rm nor}}(p,k)$ to various directions. First of all, we will allow for negative powers $-p$
with $p \in (0,1/2)$. Secondly, we would like to define the same limiting variables but associated with the stable limit $S=(S^{(1)}, S^{(2)})$ from Theorem \ref{multivariateclt} rather than $W$. Thus,    
for $d=2$, $p \in (-1/2,1/2) \setminus \{0\}$, $k_1=k_2=k$ and $r_j=j$, we set
\begin{align*}
&B_{\text{\rm low}}^{\text{\rm sta}}(p,k)=JF \left(m_{p,k}, 2^H m_{p,k}, \varphi(t_1;k), \varphi(t_2;k) \right) S^{\star}, \\[1.5 ex]
&B_{\text{\rm high}}^{\text{\rm sta}}(p,k)= \\[1.5 ex]
& \nabla v(m_{p,k}, 2^H m_{p,k}) (S^{(1)})^{\star}
\times \left(
\begin{array} {c}
\nabla G_1(\varphi(t_1;k), \varphi(t_2;k),H) \left(t_1\varphi'(t_1;k), t_2\varphi'(t_2;k),0 \right)^{\star} \\
\nabla G_2(\varphi(t_1;k), \varphi(t_2;k),H) \left(t_1\varphi'(t_1;k), t_2\varphi'(t_2;k),0 \right)^{\star} \\
1
\end{array} 
\right).
\end{align*}

\begin{rem} \label{rem1}  \rm 
In the low frequency setting there is of course no need to rely on the ratio statistic 
$R_{\text{low}}(p,k)$ to obtain an asymptotically normal estimator of the parameter $\theta=(\sigma,
\alpha, H)$. The empirical characteristic function (or, more precisely, its real part) $\varphi_{\text{low}}(t;k)_n$ is a more natural probabilistic tool for the statistical inference for $\theta$. We observe that 
a multivariate central limit theorem for the triple $(\varphi_{\text{low}}(t_j;k)_n)_{1\leq j\leq 3}$ suffices
to obtain an asymptotically normal $\sqrt{n}$-estimator $\widehat {\theta}_n$ of $\theta$. 
To increase the efficiency we may consider $l$ points $t_l >\ldots>t_1>0$  with $l \geq 3$
and minimise the asymptotic variance over $(t_1, \ldots, t_l)$.
Since the 
mathematical derivation is very similar to Theorem \ref{multivariateclt}, we leave the details to the reader.

Another intuitive method in the low frequency framework is to estimate $\theta$ via a minimal contrast approach. Given a positive weight function $w\in L^1(\R_+)$ we may obtain an estimator
$\widetilde{\theta}_n$ of $\theta$ by
\[
\widetilde{\theta}_n \in \text{argmin}_{\theta \in \R_+ \times (0,2) \times (0,1)}
\int_0^\infty \left(\varphi_{\text{low}}(t;k)_n - \varphi(t;k) \right)^2 w(t) dt.
\] 
In this setting we are likely to require tightness or a similar property of the stochastic process 
$\varphi_{\text{low}}(\cdot;k)_n$ to prove asymptotic normality of $\widetilde{\theta}_n$. However,
this seems to be a non-trivial problem, at least when using standard tightness criteria for the space 
$(C(\R_+), \| \cdot \|_{\infty})$.  We leave it for future research. \qed
\end{rem}

\begin{rem} \label{rem2} \rm
The described statistical methodology can be applied to more general processes than the mere linear
fractional stable motion. In the paper \cite{BLP} the authors investigated limit theorems for stochastic processes of the form
\[
Y_t = \int_{\R} \{g(t-s) - g_0(-s)\} dL_s,
\]
where $g,g_0$ are deterministic functions vanishing on $\R_{-}$ with $g(x)=x^{H-1/\alpha} f(x)$ and 
$f(0) \not = 0$, and $L$ is a symmetric $\alpha$-stable L\'evy motion. In the high frequency setting 
the process $Y$ exhibits the \textit{tangent process} $f(0) X$, i.e. we have that 
\[
\Delta_{i,k}^{n,r} Y \approx f(0)\Delta_{i,k}^{n,r} X.
\] 
In particular, under certain assumption on $f$ (cf. \cite{BLP}), the central limit theorem part of Theorem \ref{multivariateclt}
holds for the more general class of processes $Y$. Hence, in this semi-parametric model it is possible 
to estimate the parameter $(|f(0)| \sigma, \alpha, H)$ via the same approach as presented in Theorem \ref{th3}(ii). We remark that the function $f$ can't be inferred  from high frequency observations on a fixed time interval. \qed 
\end{rem}

\section{Statistical inference in the general  case} \label{sec4}
\setcounter{equation}{0}
\renewcommand{\theequation}{\thesection.\arabic{equation}}

In this section we treat the case of a general linear fractional stable motion as it has been introduced at 
\eqref{fLm}. We recall that in the continuous setting the restriction $H-1/ \alpha>0$ has led to the lower bound $\alpha >1$, which is essential for obtaining the asymptotic results of Theorem \ref{th3}. Without having an explicit lower bound for the stability parameter $\alpha$ statistical inference turns out to be more complex. As a consequence we will require a different estimation method for the 
self-similarity parameter $H$ and a two-step procedure to choose the right order of increments $k$. Furthermore, in order to obtain fast rates of convergence we need different treatments for the low and high frequency frameworks.

\subsection{Low frequency setting} \label{sec4.1}

We note that the basic idea behind the ratio statistic  
$R_{\text{low}}(p,k)_n$ introduced in \eqref{ratio} 
is the homogeneity of the function $f_p(x)=|x|^p$ and the fact that $m_{p,k} <\infty$
which is a consequence of $p<\alpha$ (for the associated 
central limit theorem we need the stronger condition $p < \alpha/2$). In order to keep both properties we may instead consider the negative power variation, which corresponds to the function
$f_{-p}(x)=|x|^{-p}$, and we assume throughout this section that $p\in (0,1/2)$. This approach has been 
originally proposed in \cite{DI}, although central limit theorems have not been investigated in this setting. Note that the function $f_{-p}$ is still homogenous and  $m_{-2p,k} <\infty$, which is due to the
fact that for any random variable $Y$ with  bounded density near $0$ it holds that $\E[|Y|^a] <\infty$
for all $a\in (-1,0)$.  Thus, $\widehat{H}_{\text{low}} (-p,k)_n$ is a  strongly consistent estimator of the parameter $H$ for any $p\in (0,1/2)$.

In the next step we need to ensure that we end up in the domain of attraction of the central limit theorem in Theorem
\ref{th1}(ii), which requires that $k>H+1/\alpha$. To guarantee this we need a preliminary estimator of the parameter $\alpha$. They are obtained as in \eqref{alphasigma} using the function $f_{-p}$ and
$k=1$:
\begin{align} \label{prealpha}
\widehat{\alpha}_{\text{low}}^{0} (t_1,t_2)_n
= G_2\left( \varphi_{\text{low}}(t_1)_n, 
\varphi_{\text{low}}(t_2)_n, \widehat{H}_{\text{low}} (-p)_n \right),
\end{align} 
where $G=(G_1, G_2)$.
Notice that this estimator is consistent, but we do not know if it is in the domain of attraction of a normal distribution or not. Now, we define 
\begin{align} \label{kn}
\widehat{k}_{\text{low}} (t_1,t_2)_n = 2+ 
\lfloor \widehat{\alpha}_{\text{low}}^{0} (t_1,t_2)_n^{-1}  \rfloor. 
\end{align}   
For the sake of brevity we write  $\widehat{k}_{\text{low}}= \widehat{k}_{\text{low}} (t_1,t_2)_n$. In the second step we estimate the parameter $\theta= (\sigma,
\alpha, H)$ using $\widehat{k}_{\text{low}}$. The self-similarity parameter $H$ is thus estimated by 
$\widehat{H}_{\text{low}} (-p, \widehat{k}_{\text{low}})_n$. Next, similarly to definitions at \eqref{alphasigma}, we introduce the estimators
\begin{align} \label{tildealphasigma}
&\left( \widetilde{\sigma}_{\text{low}} (\widehat{k}_{\text{low}},t_1,t_2)_n, \widetilde{\alpha}_{\text{low}} (\widehat{k}_{\text{low}},t_1,t_2)_n \right)\\[1.5 ex]
&= G\left( \varphi_{\text{low}}(t_1; \widehat{k}_{\text{low}})_n, 
\varphi_{\text{low}}(t_2; \widehat{k}_{\text{low}})_n, \widehat{H}_{\text{low}} 
(-p, \widehat{k}_{\text{low}})_n \right). \nonumber
\end{align}  
In order to determine the asymptotic distribution of the proposed estimators we will need the full force 
of Theorem \ref{multivariateclt}. Due to definition \eqref{kn} we also require a separate treatment of the cases  $\alpha^{-1} \not \in \N$
and $\alpha^{-1}  \in \N$. In the first case $\widehat{k}_{\text{low}} \toas 2+ 
\lfloor \alpha^{-1}  \rfloor$ while in the second case we will have
\[
\mathbb P\left( \widehat{k}_{\text{low}}  = 2+ 
 \alpha^{-1}  \right) \to \lambda  \qquad \text{and} \qquad  
\mathbb P\left( \widehat{k}_{\text{low}}  = 1+ 
 \alpha^{-1}  \right) \to 1- \lambda
\]
for a certain constant $\lambda \in (0,1)$. In the first setting, which is easier to treat, we obtain the following result.

\begin{theo} \label{th4}
Let $X$ be the linear fractional stable motion defined at \eqref{fLm}. Assume that $p \in (0,1/2)$ 
and $\alpha^{-1} \not \in \N$. We obtain the central limit theorem
\[
\sqrt{n} \left( 
\begin{array} {c}
\widetilde{\sigma}_{\text{\rm low}} (\widehat{k}_{\text{\rm low}},t_1,t_2)_n - \sigma \\
\widetilde{\alpha}_{\text{\rm low}} (\widehat{k}_{\text{\rm low}},t_1,t_2)_n - \alpha \\
\widehat{H}_{\text{\rm low}} (-p, \widehat{k}_{\text{\rm low}})_n - H 
\end{array}
\right)
\schw B_{\text{\rm low}}^{\text{\rm nor}} \left(-p,2+ 
\lfloor \alpha^{-1}  \rfloor \right).
\]
\end{theo}
  In the framework $\alpha^{-1}  \in \N$ we distinguish two further cases, that determine the asymptotic behaviour the preliminary estimate  $\widehat{\alpha}^0_{\text{low}}$, which is constructed using $k=1$. According to 
Theorem \ref{multivariateclt}  we are in the domain of the validity of a central limit theorem when 
$H<1-1/\alpha$ while a non-central limit theorem holds if  $H>1-1/\alpha$.

\begin{prop} \label{prop1}
Let $X$ be the linear fractional stable motion defined at \eqref{fLm}. Assume that $p \in (0,1/2)$.
\newline
(i) (Normal case) Assume that $H<1-1/\alpha$. Then we obtain 
the central limit theorem
\begin{align*}
\sqrt{n}\left( \widehat{\alpha}_{\text{\rm low}}^{0} (t_1,t_2)_n - \alpha \right)
\schw 
B_{\text{\rm low}}^{\text{\rm nor}} \left(-p,1 \right)_2.
\end{align*}
(ii) (Stable case)
Assume that  $H>1-1/\alpha$. Then we obtain 
the weak limit theorem
\begin{align*}
n^{1-1/(1+ \alpha(1-H))} \left( \widehat{\alpha}_{\text{\rm low}}^{0} (t_1,t_2)_n - \alpha \right)
\schw  B_{\text{\rm low}}^{\text{\rm sta}} \left(-p,1 \right)_2.
\end{align*}
\end{prop}
We note that the result of Proposition \ref{prop1}(ii) is essentially the same as in the asymptotically normal regime except that the convergence rate is now $n^{1-1/(1+ \alpha(1-H))}$ and the normal limit $W$ is replaced by $S$.

The next theorem presents the statistical behaviour of the estimator $(\widetilde{\sigma}_{\text{low}}, \widetilde{\alpha}_{\text{low}}, \widehat{H}_{\text{low}} 
(-p, \widehat{k}_{\text{low}})_n)$ in the case $\alpha^{-1} \in \N$.
\begin{theo} \label{th5}
Let $X$ be the linear fractional stable motion defined at \eqref{fLm}. Assume that $p \in (0,1/2)$ 
and $\alpha^{-1}  \in \N$. \newline
(i) (Case $H<1-1/\alpha$) Assume that $H<1-1/\alpha$. Then we obtain 
\[
\sqrt{n} \left( 
\begin{array} {c}
\widetilde{\sigma}_{\text{\rm low}} (\widehat{k}_{\text{\rm low}},t_1,t_2)_n - \sigma \\
\widetilde{\alpha}_{\text{\rm low}} (\widehat{k}_{\text{\rm low}},t_1,t_2)_n - \alpha \\
\widehat{H}_{\text{\rm low}} (-p,\widehat{k}_{\text{\rm low}})_n - H 
\end{array}
\right)
\schw D_{\text{\rm low}}^{\text{\rm nor}},
\]
where the probability distribution $D_{\text{\rm low}}^{\text{\rm nor}}$ on $\R^3$ is given by
\begin{align*}
D_{\text{\rm low}}^{\text{\rm nor}} (\cdot) &= \mathbb{P} (\{B_{\text{\rm low}}^{\text{\rm nor}} (-p, 2+\alpha^{-1}) 
\in \cdot\} \cap \{B_{\text{\rm low}}^{\text{\rm nor}} \left(-p,1 \right)_2<0 \}) \\
&+ \mathbb{P} (\{B_{\text{\rm low}}^{\text{\rm nor}} (-p, 1+\alpha^{-1}) 
\in \cdot\} \cap \{B_{\text{\rm low}}^{\text{\rm nor}} \left(-p,1 \right)_2>0 \}).
\end{align*}
(ii) (Case $H>1-1/\alpha$) Assume that $H>1-1/\alpha$. Then we obtain 
\[
\sqrt{n} \left( 
\begin{array} {c}
\widetilde{\sigma}_{\text{\rm low}} (\widehat{k}_{\text{\rm low}},t_1,t_2)_n - \sigma \\
\widetilde{\alpha}_{\text{\rm low}} (\widehat{k}_{\text{\rm low}},t_1,t_2)_n - \alpha \\
\widehat{H}_{\text{\rm low}} (-p,\widehat{k}_{\text{\rm low}})_n - H 
\end{array}
\right)
\schw D_{\text{\rm low}}^{\text{\rm sta}},
\]
where the probability distribution $D_{\text{\rm low}}^{\text{\rm sta}}$ on $\R^3$ is given by
\begin{align*}
D_{\text{\rm low}}^{\text{\rm sta}} (\cdot) &= \mathbb P(B_{\text{\rm low}}^{\text{\rm sta}} \left(-p,1 \right)_2<0)
\mathbb{P} (B_{\text{\rm low}}^{\text{\rm nor}} (-p, 2+\alpha^{-1}) 
\in \cdot) \\
&+ \mathbb P(B_{\text{\rm low}}^{\text{\rm sta}} \left(-p,1 \right)_2>0) \mathbb{P} (B_{\text{\rm low}}^{\text{\rm nor}} (-p, 1+\alpha^{-1}) 
\in \cdot).
\end{align*}
\end{theo}
According to Theorem \ref{multivariateclt} the statistic $(B_{\text{\rm low}}^{\text{\rm nor}} (-p, k), 
B_{\text{\rm low}}^{\text{\rm nor}} (-p, 1))$ is  jointly normal for $k \in \{1+\alpha^{-1},2+\alpha^{-1}\}$. Thus, the probability distribution  $D_{\text{\rm low}}^{\text{\rm nor}}$ can be easily computed using conditioning rules for normal distribution. 

Note however that it is problematic to use Theorem \ref{th5} for constructing confidence regions since we do not know a priori whether part (i) or part (ii) applies. We now introduce a decision rule that helps us to solve this problem.  Let $t_4>t_3>t_2>t_1>0$ be given real numbers and let 
$\widehat{\alpha}_{\text{low}}^{0} (t_1,t_2)_n$, $\widehat{\alpha}_{\text{low}}^{0} (t_3,t_4)_n$ be two estimators of parameter $\alpha \in (0,2)$ defined at \eqref{prealpha}.  Then, similarly to Proposition 
\ref{prop1}, we deduce that 
\begin{align*} 
a_n \left(\widehat{\alpha}_{\text{\rm low}}^{0} (t_3,t_4)_n - \widehat{\alpha}_{\text{\rm low}}^{0} (t_1,t_2)_n\right) \qquad \text{converges in law,}
\end{align*}  
where $a_n= \sqrt{n}$ if $H<1-1/\alpha$ and $a_n= n^{1-1/(1+ \alpha(1-H))} $ if $H>1-1/\alpha$. Hence, we immediately conclude the convergence
\begin{align*}
d_n:=-\frac{\log \left| \widehat{\alpha}_{\text{\rm low}}^{0} (t_3,t_4)_n - \widehat{\alpha}_{\text{\rm low}}^{0} (t_1,t_2)_n \right|}{\log (n)} \toop 
\begin{cases}
1/2: & \text{if } H<1-1/\alpha \\
1-1/(1+ \alpha(1-H)): & \text{if } H>1-1/\alpha
\end{cases}
\end{align*} 
In other word, the statistic $d_n$ helps us to identify the rate of convergence, but it has a bias of order $1/ \log(n)$. Our decision rule is now as follows: Use Theorem  \ref{th5}(i) to perform statistical inference if 
\[
d_n >1/2 - (\log(n))^{-1+\epsilon}
\] 
for some small chosen $\epsilon>0$; otherwise use Theorem  \ref{th5}(ii). 

\begin{rem} \rm
While we can obtain fully feasible asymptotic theory if we know whether $\alpha^{-1} \in \N$ or not, we are not yet able to deduce a complete statistical method without this a priori knowledge. Possibly subsampling procedures are required to obtain empirical confidence regions that automatically adapt 
to a given setting.
\qed
\end{rem}

\subsection{High frequency setting} \label{sec4.2}

In the framework of high frequency observations the application of the empirical characteristic function might lead to suboptimal convergence rates for the estimator of $(\sigma, \alpha)$. 
This comes from the following observation. Assume that $\alpha<1$. Using the inequality $|\cos(x)-\cos(y)| \leq |x-y|^{\alpha'}$ for any $\alpha'<\alpha$ we obtain the upper bound
\begin{align*}
&|\varphi_{\text{high}}(t; \widehat{H}_{\text{high}} (p,k)_n,k)_n - 
\varphi_{\text{high}}(t; H,k)_n|  \\[1.5 ex]
& \leq  \frac{t^{\alpha'}(n^{\widehat{H}_{\text{high}} (p,k)_n-H} -1)^{\alpha'}}{n} \sum_{i=k}^n   |n^H\Delta_{i,k}^{n} X|^{\alpha'} = O_{\mathbb P} \left( (n^{-1/2} \log n)^{-\alpha'/2}\right),
\end{align*}
where the last statement follows from $\E[|\Delta_{k,k} X|^{\alpha'}]<\infty$ and the ergodic theorem. 
Since the above expression is  predominant  in the asymptotic theory and it seems hard to improve it, we obtain slow rates of convergence for the parameters $\sigma$ and $\alpha$ if we apply the same estimation procedure as in the previous section. For this reason we require a different approach in the high frequency setting. 

First of all, we give an explicit formula for the constant $m_{-p,k}= \E[|\Delta_{k,k} X|^{-p}]$, $p \in (0,1/2)$, which has been introduced in Theorem \ref{th1}. We recall that the random variable 
$\Delta_{k,k} X$ is symmetric $\alpha$-stable with scale parameter $\sigma \|h_k\|_{\alpha}$. Consequently, applying the identity \cite[Eq. (18)]{DI} we conclude that 
\[
m_{-p,k} =  \frac{(\sigma \|h_k\|_{\alpha})^{-p}}{a_{-p}} \int_{\R} \exp(-|y|^{\alpha}) |y|^{-1+p} dy
= \frac{2(\sigma \|h_k\|_{\alpha})^{-p}}{\alpha a_{-p}} \Gamma(p/\alpha),
\] 
where the last equality follows by substitution $z=y^{\alpha}$ for $y>0$. Now, we use the idea that has been originally proposed in \cite{DI} to identify the parameter $\alpha$ via power variation statistics. We consider $p,p' \in (0,1/2)$, $p \not = p'$, and observe that
\begin{align} \label{defphi}
\frac{m_{-p',k}^{p}}{m_{-p,k}^{p'}} = 
\frac{(2/\alpha)^{p-p'} a_{-p}^{p'} \Gamma(p'/\alpha)^p }{a_{-p'}^{p} \Gamma(p/\alpha)^{p'}}=: \phi_{p,p'}(\alpha). 
\end{align}
It has been shown in \cite{DI} that the mapping $\alpha \mapsto \phi_{p,p'}(\alpha)$ is invertible for any 
$p \not = p'$. Hence, we have $\alpha= \phi_{p,p'}^{-1}(m_{-p',k}^{p}/ m_{-p,k}^{p'})$. Now, assuming that we know $\alpha$ and $H$ (recall that the norm $ \|h_k\|_{\alpha}$ depends on these parameters), we can recover the scale parameter $\sigma$ via
\[
\sigma= \left( \frac{\alpha a_{-p} m_{-p,k}}{2  \Gamma(p/\alpha)} \right)^{- \frac 1p} /  \|h_k\|_{\alpha}. 
\]
Summarising the above identities we obtain the function $\overline G: (\R_+)^2 \times (0,1) 
\to \R^2$ such that  
\begin{align} \label{Gbar}
(\sigma, \alpha ) = \overline{G} \left(m_{-p,k}, m_{-p',k}, H \right).
\end{align} 
Next, we follow the same two-stage routine as in the previous section. We first compute  
$\widehat{H}_{\text{high}} (-p)_n =\widehat{H}_{\text{high}} (-p,1)_n$ with $p \in (0,1/2)$ and define the preliminary estimator of 
$\alpha$ by 
\begin{align} \label{prealphahiogh}
\widehat{\alpha}_{\text{high}}^{0} (-p,-p')_n
= \overline{G}_2\left( V_{\text{high}}(f_{-p}, \widehat{H}_{\text{high}} (-p)_n)_n, 
V_{\text{high}}(f_{-p'}, \widehat{H}_{\text{high}} (-p)_n)_n, \widehat{H}_{\text{high}} (-p)_n \right), 
\end{align} 
where the statistic $V_{\text{high}}(f_{-p}, \widehat{H}_{\text{high}} (-p)_n)_n$ refers to power variation introduced in \eqref{Vstat} with $k=1$ and with $H$ replaced by $\widehat{H}_{\text{high}} (-p)_n$.  
Now, we define 
\begin{align}  \label{khigh}
\widehat{k}_{\text{high}}= \widehat{k}_{\text{high}} (-p,-p')_n &= 2+ 
\lfloor \widehat{\alpha}_{\text{high}}^{0} (-p,-p')_n^{-1}  \rfloor 
\end{align}   
and introduce the estimator
\begin{align*} 
&\left( \widetilde{\sigma}_{\text{high}} (\widehat{k}_{\text{high}},-p,-p')_n, \widetilde{\alpha}_{\text{high}} (\widehat{k}_{\text{high}},-p,-p')_n \right)
= \overline{G}\left( V_{\text{high}}(f_{-p}, \widehat{H}_{\text{high}} (-p, \widehat{k}_{\text{high}})_n; \widehat{k}_{\text{high}})_n,  \right. \\
&\left. 
V_{\text{high}}(f_{-p'}, \widehat{H}_{\text{high}} (-p, \widehat{k}_{\text{high}})_n; \widehat{k}_{\text{high}})_n, \widehat{H}_{\text{high}} (-p, \widehat{k}_{\text{high}})_n \right).
\end{align*} 
We again require a separate treatment of the cases $\alpha^{-1} \not \in \N$ and $\alpha^{-1} \in \N$. We start with the first setting. When $H<k-1/\alpha$ we consider the statistic $W(n)^{(1)}= (W(n)^{(1)}_1, W(n)^{(1)}_2)$ associated with the power $-p$ and
\[
 k_1 = \widehat{k}_{\text{high}}, r_1=1 \quad \text{and} \quad  k_2 = \widehat{k}_{\text{high}}, r_2=2.
\]
Recall that $W(n)^{(1)} \schw W^{(1)}$ according to Theorem \ref{th1}. Now, similarly to Theorem 
\ref{th3}, we define  
\begin{align} \label{Bover}
\overline{B}_{\text{\rm high}}^{\text{\rm nor}}(-p,-p',k)&:= \nabla v_p(m_{-p,k}, 2^H m_{-p,k}) (W^{(1)})^{\star} 
\\[1.5 ex]
& 
\times \left(
\begin{array} {c}
\nabla \overline{G}_1(m_{-p,k}, m_{-p',k}, H) \left(-pm_{-p,k}, -p'm_{-p',k}, H  \right)^{\star} \\
\nabla \overline{G}_1(m_{-p,k}, m_{-p',k}, H) \left(-pm_{-p,k}, -p'm_{-p',k}, H  \right)^{\star}  \\
1
\end{array} 
\right), \nonumber
\end{align}
where the function $v_p$ has been introduced  at \eqref{vF}. Our first result is the following theorem.

\begin{theo} \label{th6}
Let $X$ be the linear fractional stable motion defined at \eqref{fLm}. Assume that $p,p' \in (0,1/2)$ 
and $\alpha^{-1} \not \in \N$. Then we obtain the central limit theorem
\begin{align*}
&\left( 
\begin{array} {c}
\sqrt{n} (\log n)^{-1} \left(\widetilde{\sigma}_{\text{\rm high}} (\widehat{k}_{\text{\rm high}},-p,-p')_n - \sigma \right) \\
\sqrt{n} (\log n)^{-1} \left( \widetilde{\alpha}_{\text{\rm high}} (\widehat{k}_{\text{\rm high}},-p,-p')_n  - \alpha \right) \\
\sqrt{n}  \left(\widehat{H}_{\text{\rm high}} (-p,\widehat{k}_{\text{\rm high}})_n - H \right) 
\end{array}
\right)
\schw \overline{B}_{\text{\rm high}}^{\text{\rm nor}}\left(-p,-p',2+ 
\lfloor \alpha^{-1}  \rfloor \right).  
\end{align*}  
\end{theo}
Next, we  treat the case $\alpha^{-1} \in \N$. For this purpose, whenever $H>k-1/\alpha$, we introduce the notation  
$\overline{B}_{\text{\rm high}}^{\text{\rm sta}}(-p,-p',k)$ to denote the random variable at \eqref{Bover}
where $W^{(1)}$ is replaced by $S^{(1)}$. We deduce the following result, which is the analogue of Theorem \ref{th5}. 

\begin{theo} \label{th7}
Let $X$ be the linear fractional stable motion defined at \eqref{fLm}. Assume that $p,p' \in (0,1/2)$ 
and $\alpha^{-1}  \in \N$. \newline
(i) (Case $H<1-1/\alpha$) Assume that $H<1-1/\alpha$. Then we obtain 
\begin{align*}
\left( 
\begin{array} {c}
\sqrt{n} (\log n)^{-1} \left(\widetilde{\sigma}_{\text{\rm high}} (\widehat{k}_{\text{\rm high}},-p,-p')_n - \sigma \right) \\
\sqrt{n} (\log n)^{-1} \left( \widetilde{\alpha}_{\text{\rm high}} (\widehat{k}_{\text{\rm high}},-p,-p')_n  - \alpha \right) \\
\sqrt{n}  \left(\widehat{H}_{\text{\rm high}} (-p,\widehat{k}_{\text{\rm high}})_n - H \right) 
\end{array}
\right)
\schw D_{\text{\rm high}}^{\text{\rm nor}},
\end{align*}
where the probability distribution $D_{\text{\rm high}}^{\text{\rm nor}}$ on $\R^3$ is given by
\begin{align*}
D_{\text{\rm high}}^{\text{\rm nor}} (\cdot) &= \mathbb{P} (\{\overline B_{\text{\rm high}}^{\text{\rm nor}} (-p, -p',2+\alpha^{-1}) 
\in \cdot\} \cap \{\overline B_{\text{\rm high}}^{\text{\rm nor}} \left(-p, -p',1 \right)_2<0 \}) \\
&+ \mathbb{P} (\{ \overline B_{\text{\rm high}}^{\text{\rm nor}} (-p,-p', 1+\alpha^{-1}) 
\in \cdot\} \cap \{\overline B_{\text{\rm high}}^{\text{\rm nor}} \left(-p, -p',1 \right)_2>0 \}).
\end{align*}
(ii) (Case $H>1-1/\alpha$) Assume that $H>1-1/\alpha$. Then we obtain  
\begin{align*}
\left( 
\begin{array} {c}
\sqrt{n} (\log n)^{-1} \left(\widetilde{\sigma}_{\text{\rm high}} (\widehat{k}_{\text{\rm high}},-p,-p')_n - \sigma \right) \\
\sqrt{n} (\log n)^{-1} \left( \widetilde{\alpha}_{\text{\rm high}} (\widehat{k}_{\text{\rm high}},-p,-p')_n  - \alpha \right) \\
\sqrt{n}  \left(\widehat{H}_{\text{\rm high}} (-p,\widehat{k}_{\text{\rm high}})_n - H \right) 
\end{array}
\right)
\schw D_{\text{\rm high}}^{\text{\rm sta}},
\end{align*}
where the probability distribution $D_{\text{\rm high}}^{\text{\rm sta}}$ on $\R^3$ is given by
\begin{align*}
D_{\text{\rm high}}^{\text{\rm sta}} (\cdot) &= \mathbb P(\overline B_{\text{\rm high}}^{\text{\rm sta}} \left(-p,-p',1 \right)_2<0)
\mathbb{P} (\overline B_{\text{\rm high}}^{\text{\rm nor}} (-p,-p', 2+\alpha^{-1}) 
\in \cdot) \\
&+ \mathbb P(\overline B_{\text{\rm high}}^{\text{\rm sta}} \left(-p,-p',1 \right)_2>0) \mathbb{P} (\overline B_{\text{\rm high}}^{\text{\rm nor}} (-p, -p',1+\alpha^{-1}) 
\in \cdot).
\end{align*}
\end{theo}

\begin{rem} \rm
We may use a similar decision rule as proposed in Section \ref{sec4.1} to figure out whether part (i) or (ii) of Theorem \ref{th7} is applicable. Let $p_1,\ldots, p_4 \in (0,1/2)$ be distinct  real numbers. As in the previous subsection we have that 
\begin{align*}
\overline{d}_n:=-\frac{\log \left| \widehat{\alpha}_{\text{high}}^{0} (-p_1,-p_2)_n - \widehat{\alpha}_{\text{high}}^{0} (-p_3,-p_4)_n \right|}{\log (n)} \toop 
\begin{cases}
1/2: & \text{if } H<1-1/\alpha \\
1-1/(1+ \alpha(1-H)): & \text{if } H>1-1/\alpha
\end{cases}
\end{align*} 
We thus use Theorem \ref{th7}(i) to perform statistical inference when 
\[
\overline d_n >1/2 - (\log(n))^{-1+\epsilon}.
\] 
\qed
\end{rem}

\section{A simulation study} \label{sec5}
\setcounter{equation}{0}
\renewcommand{\theequation}{\thesection.\arabic{equation}}

In this section we demonstrate the finite sample performance of our estimators based upon the theoretical results of Theorems \ref{th3}, \ref{th4} and \ref{th6}, where the latter two correspond to the setting $\alpha^{-1} \not \in \N$ (we dispense with the numerical analysis associated with Theorems \ref{th5} and \ref{th7}). We simulate high and low frequency observations of the linear fractional stable motion defined at \eqref{fLm} for $n=100$, $1.000$ and $10.000$. Whenever we use the statistics $V_{\text{high}}(f; k,r)_n$ and $V_{\text{low}}(f; k,r)_n$ introduced in \eqref{Vstat}, we multiply them by $(n-rk+1)/n$ to account for the actual number of summands. Throughout the section we set $t_1=1$ and $t_2=2$. We use 5000 repetitions  to uncover the finite sample properties of our estimators. The asymptotic variances appearing in central limit theorems are rather hard to compute numerically, so we perform Monte Carlo simulations to estimate them. For each generated sample path we compute $(\widehat{\sigma}, \widehat{\alpha}, \widehat{H})$. Then  we calculate sample mean and standard deviation, which are used to construct empirical density functions.

\begin{table}[h!]
\centering
\caption{\footnotesize Bias/standard deviation of the estimators $(\widehat{\sigma}_{\text{\rm low}}, \widehat{\alpha}_{\text{\rm low}}, \widehat{H}_{\text{\rm low}})$ and $(\widehat{\sigma}_{\text{\rm high}}, \widehat{\alpha}_{\text{\rm high}}, \widehat{H}_{\text{\rm high}})$. We use $p=0.4$ and $k=2$, and the true parameter is $(\sigma, \alpha, H) = (0.3,1.8,0.8)$. } 
\vspace{0.2 cm}
\begin{tabular}{c|ccc|ccc|}
$n$&   $\widehat{\sigma}_{\text{\rm low}}$ 	& $\widehat{\alpha}_{\text{\rm low}}$ 	& 
$\widehat{H}_{\text{\rm low}}$ 	& $\widehat{\sigma}_{\text{\rm high}}$ 	& $\widehat{\alpha}_{\text{\rm high}}$ 	& 
$\widehat{H}_{\text{\rm high}}$\\[1mm]
\hline\hline
100	    & -0.024/0.06	& -0.038/0.18     & -0.05/0.12    	&  0.06/0.18	    & -0.07/0.2	      &  0.02/0.10 \\
1000	& -0.0008/0.02	&  0.012/0.068	& -0.012/0.05	    & -0.001/0.12	&  0.015/0.07	  & -0.009/0.05  \\
10000	&  0.00014/0.006	&  0.0005/0.022	& -0.005/0.016	    & -0.010/0.05	&  0.001/0.022   & -0.005/0.016  \\
\end{tabular}
\label{table1}
\end{table}

\begin{figure}[h!]
	\centering
	\includegraphics[width=1\textwidth]{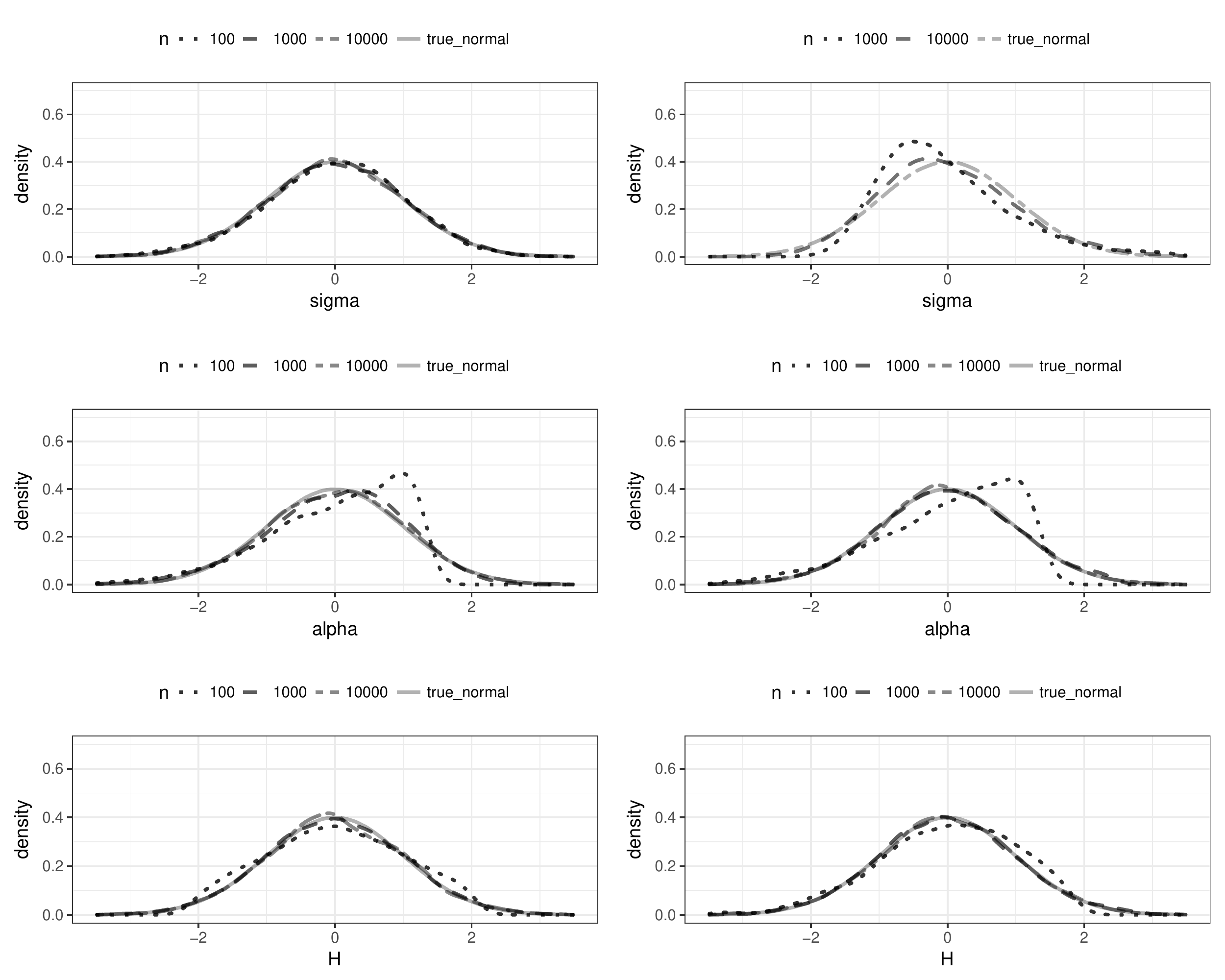}
	\caption{ \footnotesize
	        Empirical pdfs of 
	        $(\widehat{\sigma}, \widehat{\alpha}, \widehat{H})$ in high and low frequency settings. 
	        The right column corresponds to the high frequency case and the left one to the low frequency case. 
	        The true parameter is $(\sigma, \alpha, H) = (0.3,1.8,0.8)$, $k=2$, $p=0.4$.
            }
	\label{dist_3_1_HL}
\end{figure}

We begin with the discussion  of Theorem \ref{th3}. Table \ref{table1} reports the bias and the standard deviation of the estimator 
of $(\sigma, \alpha, H)=(0.3,1.8,0.8)$ in high and low frequency settings, where we use the power $p=0.4$ and the order $k=2$. We observe
that our estimators exhibit a rather convincing finite sample performance in both settings. As expected from the theoretical statements of Theorem \ref{th3}, the estimators of the self-similarity parameter $H$  exhibit similar finite sample properties in high and low frequency settings, while the
performance of the low frequency estimators for the parameters $\sigma$ and $\alpha$ is better than in the high frequency case. This is obviously a consequence of a slightly slower convergence rate in the high frequency setting. 
Figure \ref{dist_3_1_HL} plots the empirical densities
of the standardised estimators from Theorem \ref{th3} in comparison to the density of the standard normal distribution. As 
mentioned earlier we use Monte Carlo simulations to estimate the theoretical variances. 
We again observe a very good performance of estimators of the parameter $H$, while the numerical results for the
estimators of $\sigma$ and $\alpha$ are better in the low frequency case.

Now, we turn our attention to the low frequency estimation discussed in Theorem \ref{th4}. We use the power $p=-0.4$ and consider the true 
parameter $(\sigma, \alpha, H)=(0.3,1.8,0.8)$ and $(\sigma, \alpha, H)=(0.3,0.8,0.8)$. Observe that the first case corresponds to the setting of Theorem \ref{th3} and 
the second parameter corresponds
to the discontinuous setting. The estimated order $\widehat{k}_{\text{low}}$ is computed via \eqref{kn}.  Table \ref{table2} displays the bias and
standard deviation in the case $(\sigma, \alpha, H)=(0.3,1.8,0.8)$, while Table \ref{table3} 
demonstrates the numerical results   in the case $(\sigma, \alpha, H)=(0.3,0.8,0.8)$.

\begin{table}[h!]
\centering
\caption{\footnotesize 
Bias/standard deviation of the estimator $(\widetilde{\sigma}_{\text{\rm low}}, \widetilde{\alpha}_{\text{\rm low}}, \widetilde{H}_{\text{\rm low}})$. Here $p=-0.4$, $\widehat{k}_{\text{low}}$ is computed
from  \eqref{kn}  and $(\sigma, \alpha, H) = (0.3,1.8,0.8)$.} 
\vspace{0.2 cm}
\begin{tabular}{c|ccc|}
$n$&   $\widetilde{\sigma}_{\text{\rm low}}$ 	& $\widetilde{\alpha}_{\text{\rm low}}$ 	& 
$\widetilde{H}_{\text{\rm low}}$ 	\\[1mm]
\hline\hline
100	    & -0.05/0.09	& -0.031/0.18    & -0.12/0.23  \\
1000	& -0.004/0.04	& 0.01/0.068    & -0.018/0.12  \\
10000	& 0.0003/0.015	& 0.001/0.022	& -0.003/0.05  \\
\end{tabular}
\label{table2}
\end{table}

\begin{table}[h!]
	\centering
	\caption{\footnotesize Bias/standard deviation of the estimator $(\widetilde{\sigma}_{\text{\rm low}}, \widetilde{\alpha}_{\text{\rm low}}, \widetilde{H}_{\text{\rm low}})$. Here $p=-0.4$, $\widehat{k}_{\text{low}}$ is computed
from  \eqref{kn}  and $(\sigma, \alpha, H) = (0.3,0.8,0.8)$.} 
	\vspace{0.2 cm}
	\begin{tabular}{c|ccc|}
		$n$&   $\widetilde{\sigma}_{\text{\rm low}}$ 	& $\widetilde{\alpha}_{\text{\rm low}}$ 	& 
		$\widetilde{H}_{\text{\rm low}}$ 	\\[1mm]
		\hline\hline
		100	    & -0.06/0.31	& -0.003/0.41     & -0.15/0.24  \\
		1000	& -0.05/0.27	& -0.08/0.31     & 0.003/0.13  \\
		10000	&  0.03/0.26	&  0.008/0.27     & 0.04/0.05  \\
	\end{tabular}
	\label{table3}
\end{table}

\noindent Comparing the simulation results of Theorems  \ref{th3} and \ref{th4}, we see that the finite sample performance of estimators $\sigma$ and $H$ in 
Theorem \ref{th4} is inferior. This is not really surprising, since the methodology of Theorem \ref{th4} requires preliminary estimation of $\alpha$ and $k$, and hence leads to an accumulation of errors. On the other hand, the estimator of $\alpha$ is not as sensitive to preliminary estimation. Furthermore, in the setting of a fractional Brownian motion it is well known that low values of the parameter $k$ give more efficient estimators. We conjecture that a similar effect appears for linear fractional stable motions. This would explain the superiority of the results 
in Table  \ref{table2}  compared to those in Table  \ref{table3}, since $\lfloor \alpha^{-1} \rfloor +2 = 2$ in the first setting while  
$\lfloor \alpha^{-1} \rfloor +2 = 3$ in the second setting. Figures \ref{dist_4_1_d_pos} and \ref{dist_4_1_d_neg} show the empirical density
functions, where the theoretical variances have been estimated via a Monte Carlo simulations.  They confirm the better performance of the 
estimators in the continuous setting  $(\sigma, \alpha, H)=(0.3,1.8,0.8)$. We also observe that the estimator of the parameter $\sigma$ exhibits 
the worst finite sample properties in the setting $(\sigma, \alpha, H)=(0.3,0.8,0.8)$.

\begin{figure}[h!]
	\centering
	\includegraphics[width=1\textwidth]{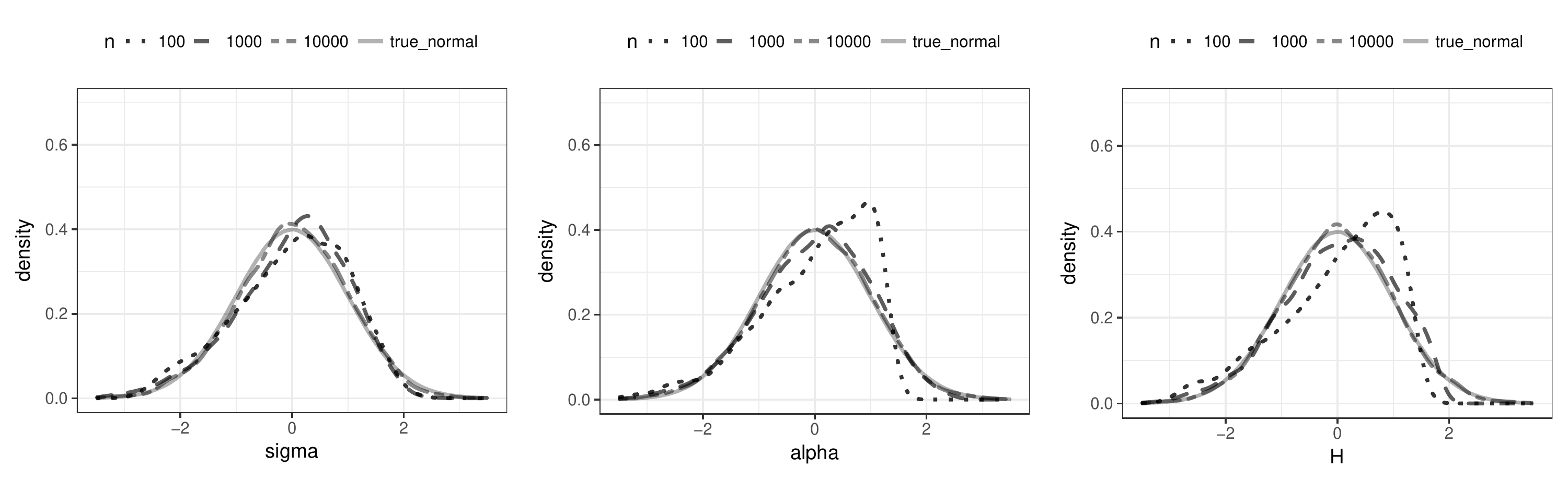}
	\caption{ \footnotesize
		Empirical pdfs of 
		$(\widetilde{\sigma}_{\text{\rm low}}, \widetilde{\alpha}_{\text{\rm low}}, \widetilde{H}_{\text{\rm low}})$. Here $(\sigma, \alpha, H) = (0.3,1.8,0.8)$ and $p=-0.4$.
	}
	\label{dist_4_1_d_pos}
\end{figure}

\begin{figure}[h!]
	\centering
	\includegraphics[width=1\textwidth]{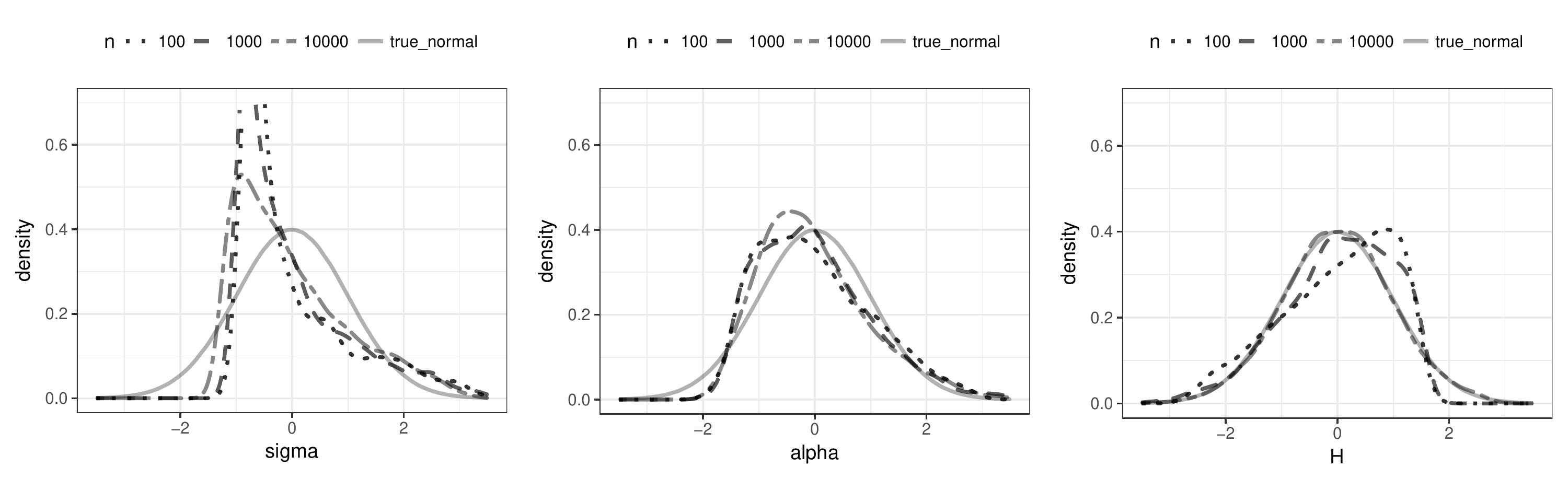}
	\caption{ \footnotesize
		Empirical pdfs of 
		$(\widetilde{\sigma}_{\text{\rm low}}, \widetilde{\alpha}_{\text{\rm low}}, \widetilde{H}_{\text{\rm low}})$. Here 
		$(\sigma, \alpha, H) = (0.3,0.8,0.8)$ and $p=-0.4$.
	}
	\label{dist_4_1_d_neg}
\end{figure}

Finally, let us discuss the finite sample performance of the high frequency estimators from Theorem \ref{th6}.  We again consider two parameter settings $(\sigma, \alpha, H)=(0.3,1.8,0.8)$ and $(\sigma, \alpha, H)=(0.3,0.8,0.8)$, and we use $p=-0.4$ and $p'=-0.2$. 
The estimated order $\widehat{k}_{\text{high}}$ is computed via \eqref{khigh}. Tables \ref{table4} and \ref{table5} display the biases and standard deviations in both parameter settings. We observe that the estimators of the parameter $\sigma$ have the worst performance and we only obtain reasonable results for $n=10.000$.  Similar conclusions can be drawn from Figures \ref{dist_4_5_d_pos} and \ref{dist_4_5_d_neg} that plot the empirical density functions. The bad performance of the estimator of $\sigma$ in Theorem \ref{th6} is explained by the fact that we not only require a preliminary estimation step for our procedure, but we also need to estimate the parameters $H$ and $\alpha$ first to obtain an estimator of $\sigma$. This leads to accumulation of finite sample errors, which results in large bias and variance for small $n$.  To further highlight this issue, we have plotted the empirical densities for the estimators of $\sigma$ from Theorems \ref{th4} and \ref{th6}  in Figure \ref{last} 
in the setting $(\sigma, \alpha, H)=(0.3,0.8,0.8)$ where the parameter $(\alpha,H)$ is assumed to be known. We observe a much better finite sample performance, which confirms that the bad finite sample properties of the estimator of $\sigma$ are largely due to preliminary estimation 
of  $(\alpha,H)$.

\newpage

\begin{table}[h!]
	\centering
	\caption{\footnotesize Bias/standard deviation of the estimator $(\widetilde{\sigma}_{\text{\rm high}}, \widetilde{\alpha}_{\text{\rm high}}, \widetilde{H}_{\text{\rm high}})$. Here $p=-0.4, p'=-0.2$ and  $(\sigma, \alpha, H) = (0.3,1.8,0.8)$.} 
	\vspace{0.2 cm}
	\begin{tabular}{c|ccc|}
		$n$&   $\widetilde{\sigma}_{\text{\rm high}}$ 	& $\widetilde{\alpha}_{\text{\rm high}}$ 	& 
		$\widetilde{H}_{\text{\rm high}}$ 	\\[1mm]
		\hline\hline
		100	    & 60/1443    	& -0.02/0.77   & 0.23/0.33  \\
		1000	& 0.18/0.82	    & 0.19/0.67     & 0.02/0.13  \\
		10000	& -0.003/0.17	& 0.052/0.26    & -0.003/0.05  \\
	\end{tabular}
	\label{table4}
\end{table}

\begin{table}[h!] 
	\centering
	\caption{\footnotesize Bias/standard deviation of the estimator $(\widetilde{\sigma}_{\text{\rm high}}, \widetilde{\alpha}_{\text{\rm high}}, \widetilde{H}_{\text{\rm high}})$. Here $p=-0.4, p'=-0.2$ and $(\sigma, \alpha, H) = (0.3,0.8,0.8)$.} 
	\vspace{0.2 cm}
	\begin{tabular}{c|ccc|}
		$n$&   $\widetilde{\sigma}_{\text{\rm low}}$ 	& $\widetilde{\alpha}_{\text{\rm low}}$ 	& 
		$\widetilde{H}_{\text{\rm low}}$ 	\\[1mm]
		\hline\hline
		100	    & 16/341 	    & 0.19/0.37     & 0.13/0.4  \\
		1000	& 0.103/1	    & 0.02/0.09     & 0.06/0.16  \\
		10000	& -0.11/0.12  	& 0.003/0.04    & 0.04/0.06  \\
	\end{tabular}
	\label{table5}
\end{table}


\begin{figure}[h!]
	\centering
	\includegraphics[width=1\textwidth]{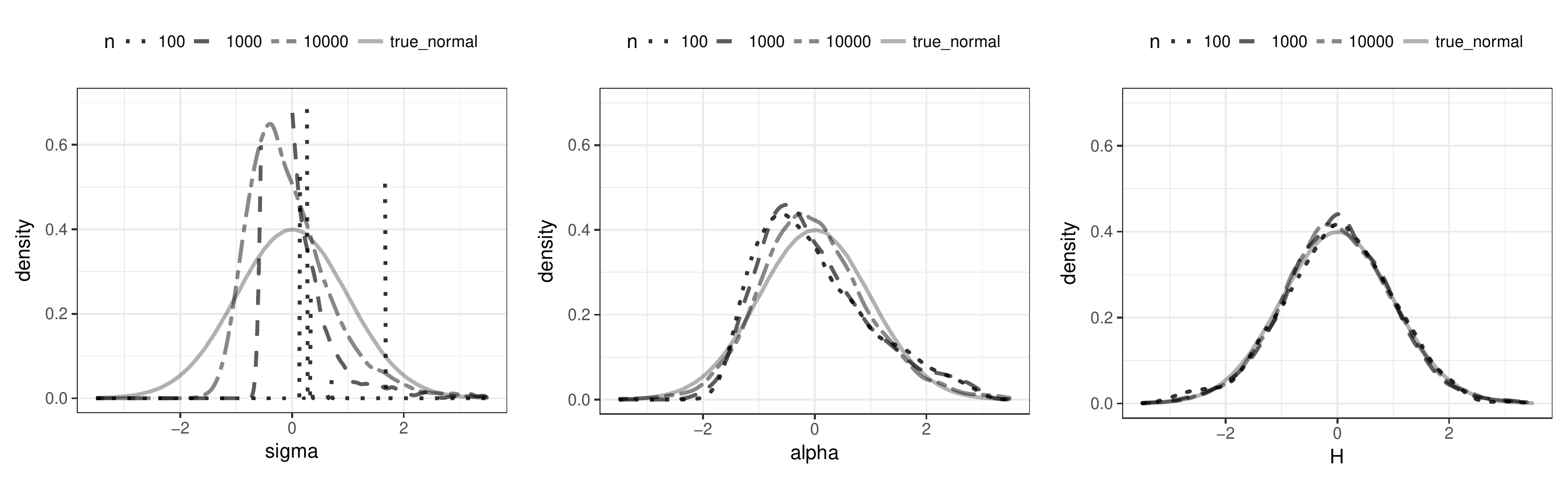}
	\caption{ \footnotesize
		Empirical pdfs of 
		$(\widetilde{\sigma}_{\text{\rm high}}, \widetilde{\alpha}_{\text{\rm high}}, \widetilde{H}_{\text{\rm high}})$. Here $p=-0.4, p'=-0.2$ and  $(\sigma, \alpha, H) = (0.3,1.8,0.8)$.
	}
	\label{dist_4_5_d_pos}
\end{figure}

\begin{figure}[h!]
	\centering
	\includegraphics[width=1\textwidth]{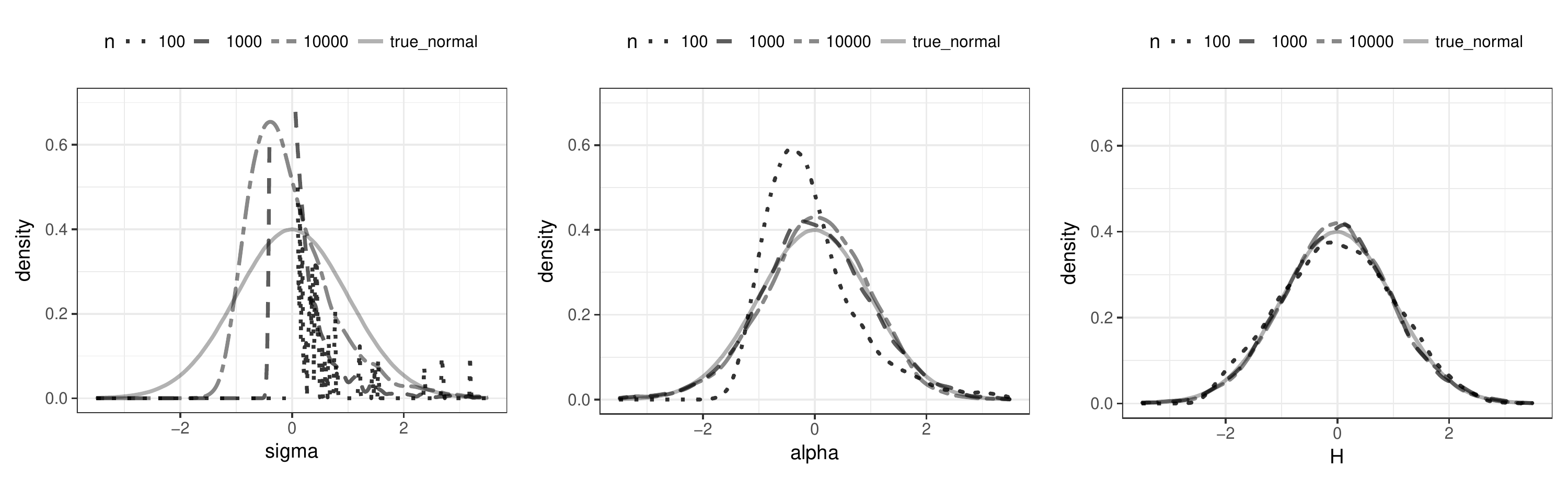}
	\caption{ \footnotesize
		Empirical pdfs of 
		$(\widetilde{\sigma}_{\text{\rm high}}, \widetilde{\alpha}_{\text{\rm high}}, \widetilde{H}_{\text{\rm high}})$. 
		Here $p=-0.4, p'=-0.2$ and $(\sigma, \alpha, H) = (0.3,0.8,0.8)$.
	}
	\label{dist_4_5_d_neg}
\end{figure}

\begin{figure}[h!]
  \centering
  \begin{subfigure}[b]{0.45\linewidth}
    \includegraphics[width=\linewidth]{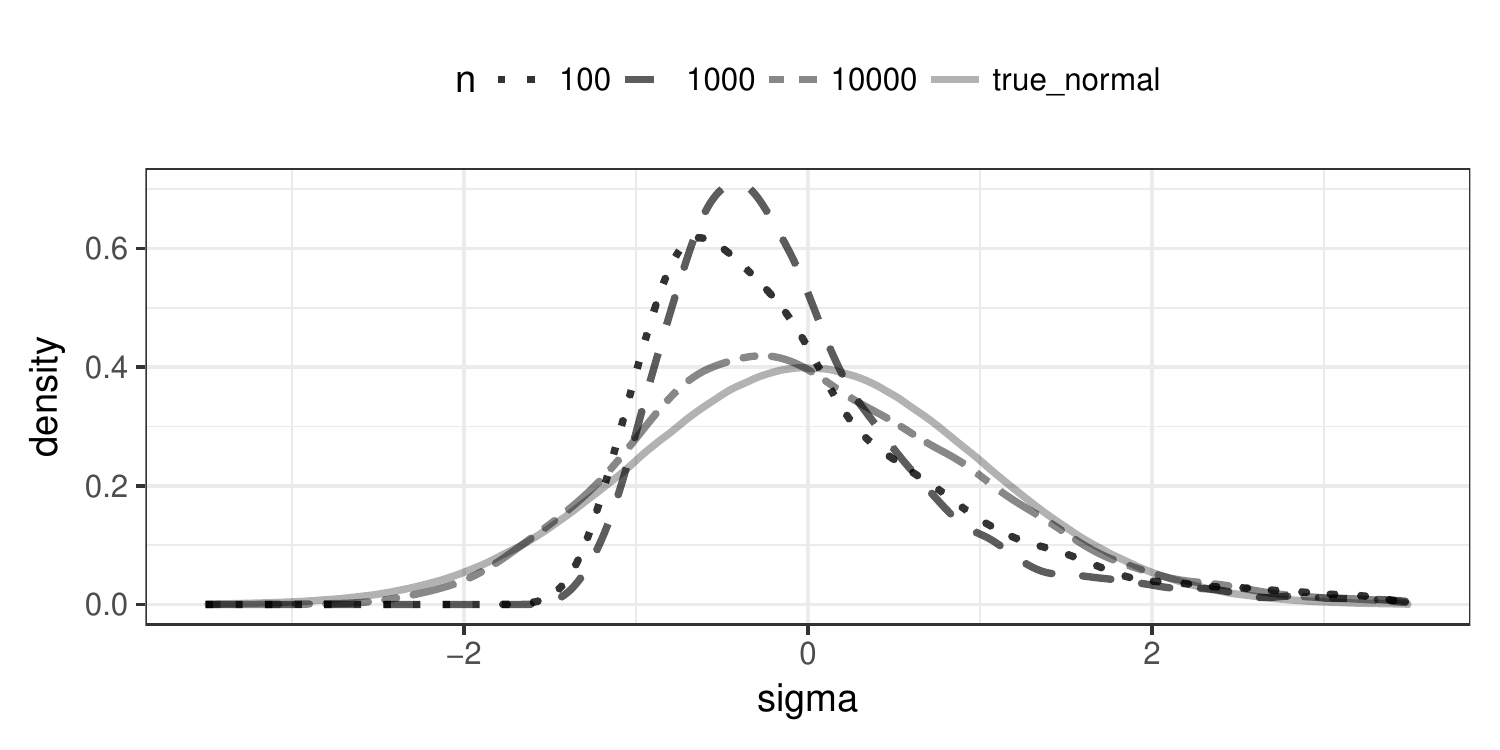}
  \end{subfigure}
  \begin{subfigure}[b]{0.45\linewidth}
    \includegraphics[width=\linewidth]{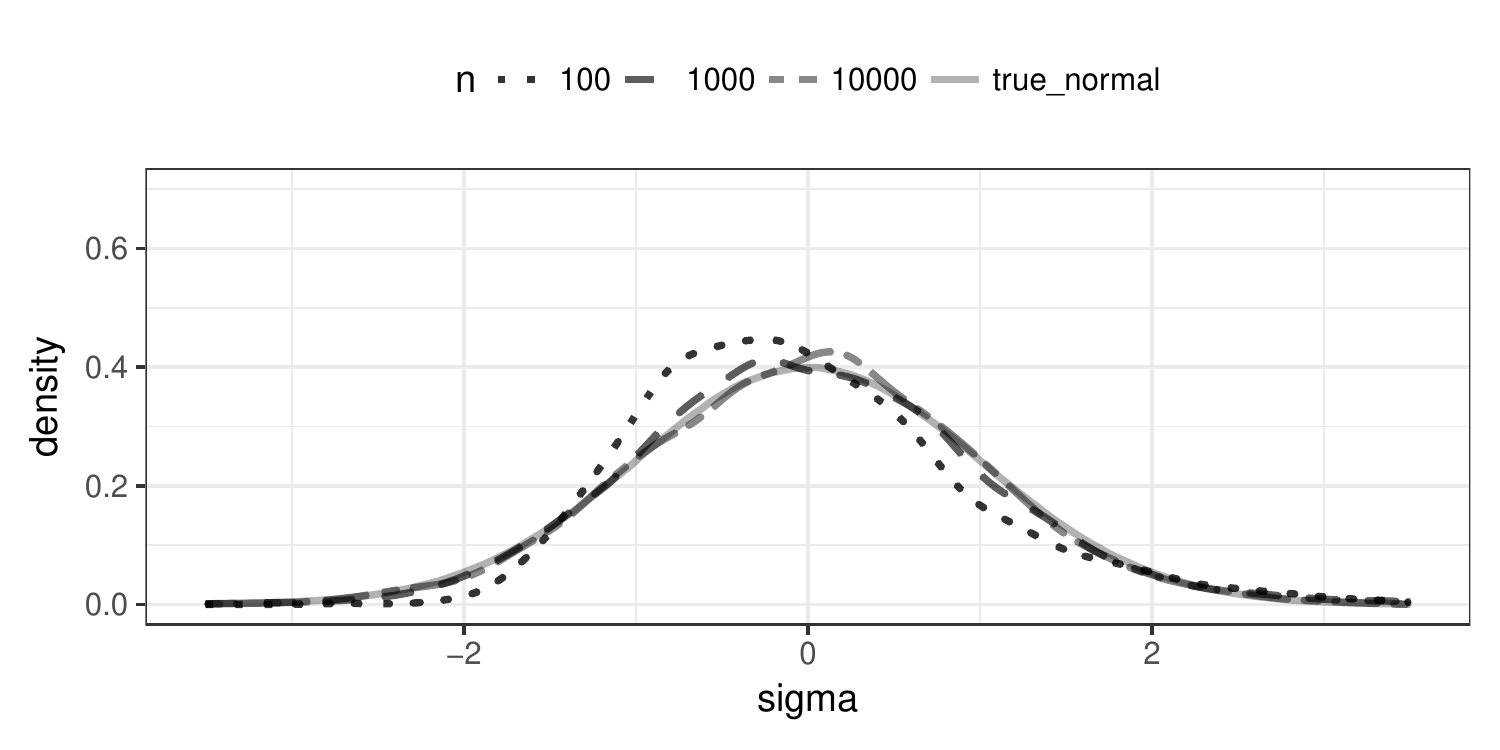}
     \end{subfigure}
 \caption{\footnotesize Empirical pdfs for $\sigma$ (Left=Theorem \ref{th4}, Right=Theorem \ref{th6}) 
 when the parameter $(\alpha,H)= (0.8,0.8)$ is known. Here $\sigma=0.3$, $p=-0.4, p'=-0.2$ and $k=3$.}
\label{last}
\end{figure}

\subsection*{Acknowledgment}
The authors acknowledge financial support from the project 
``Ambit fields: probabilistic properties and statistical inference'' funded by Villum Fonden 
and from
CREATES funded by the Danish
National Research Foundation.

\newpage

\section{Proofs} \label{sec6}
\setcounter{equation}{0}
\renewcommand{\theequation}{\thesection.\arabic{equation}}

In this section we denote all positive constants by $C$ although they may change from line to line.

\subsection{Preliminaries} \label{sec6.1}
Here we will show some technical results, which are necessary to prove the main theorems. We start with the following lemma that is a straightforward consequence of Taylor expansion. 

\begin{lem} \label{lem1}
Let $h_{k,r}$ be defined as in \eqref{def-h}. Then it holds that 
\begin{align*}
|h_{k,r}(x)| \leq C\left( x^{H-1/\alpha} 1_{(0,rk+1]}(x) + x^{H-k-1/\alpha} 1_{(rk+1,\infty)}(x)\right).
\end{align*}
Furthermore, the function $|h_{k,r}|$ is strictly decreasing on $(rk+1,\infty)$. 
\end{lem} 
An important quantity when considering various asymptotic covariances is the following object:
\begin{align} \label{rho}
\rho_l:= \int_{0}^{\infty} |h_{k,r}(x) h_{k,r}(x+l)|^{\alpha/2} dx. 
\end{align}
The next lemma determines the asymptotic behaviour of $\rho_l$ when $l\to \infty$. 

\begin{lem} \label{lem2}
For $l>rk$ it holds that 
\begin{align*} 
&\rho_l \leq C
\begin{cases}
l^{(\alpha(H-k) -1)/2} : & \text{when } k> H+1/\alpha \\
l^{\alpha(H-k)} :& \text{when } k< H+1/\alpha
\end{cases}
\end{align*}
\end{lem}
\begin{proof} 
Assume that $l>rk$. 
Applying Lemma \ref{lem1} we obtain the inequality
\begin{align*}
&\int_0^l |h_{k,r}(x) h_{k,r}(x+l)|^{\alpha/2} dx \leq C  l^{(\alpha(H-k) -1)/2} 
\int_0^l |h_{k,r}(x) |^{\alpha/2} dx  \\[1.5 ex]
& \leq C
\begin{cases}
  l^{(\alpha(H-k) -1)/2} \int_0^{\infty} |h_{k,r}(x) |^{\alpha/2} dx: & k> H+1/\alpha \\
  l^{\alpha(H-k)}:  & k< H+1/\alpha
\end{cases}
\end{align*}
When $k> H+1/\alpha$ we have $\int_0^{\infty} |h_{k,r}(x) |^{\alpha/2} dx< \infty$, which is due to
Lemma \ref{lem1}; on the other hand, for $k< H+1/\alpha$ we deduce that 
$\int_0^l |h_{k,r}(x) |^{\alpha/2} dx \leq C l^{1+(\alpha(H-k) -1)/2}$.   
Applying Lemma \ref{lem1} once again and using the substitution $x=ly$ we deduce the inequality
\begin{align*}
\int_l^{\infty} |h_{k,r}(x) h_{k,r}(x+l)|^{\alpha/2} dx &\leq C \int_l^{\infty} |x(x+l)|^{(\alpha(H-k) -1)/2} dx \\
&=C l^{\alpha(H-k)} \int_1^{\infty}  |y(y+1)|^{(\alpha(H-k) -1)/2} dy.
\end{align*}
Indeed, the last integral is finite since $H<1\leq k$. Hence, the statement of Lemma  \ref{lem2} is proved.  
\end{proof}
In the next step we will determine the behaviour of the function $U_{g,h}$ defined at \eqref{U}. The following result is the statement of inequalities (3.4)-(3.6) from \cite{PTA}. 

\begin{lem} \label{lem3}
For any $u, v \in \R$ it holds that 
\begin{align*}
|U_{g,h} (u,v)| &\leq  2|uv|^{\alpha/2} \int_{0}^{\infty} |g(x) h(x)|^{\alpha/2} dx \\[1.5 ex]
& \times \exp\left( -2|uv|^{\alpha/2} \left(\|g\|_{\alpha}^{\alpha/2} \|h\|_{\alpha}^{\alpha/2}
-  \int_{0}^{\infty} |g(x) h(x)|^{\alpha/2} dx \right)\right), \\[1.5 ex]
|U_{g,h} (u,v)| &\leq  2|uv|^{\alpha/2} \int_{0}^{\infty} |g(x) h(x)|^{\alpha/2} dx \\[1.5 ex]
& \times \exp\left( - \left(\|ug\|_{\alpha}^{\alpha/2} -  \|vh\|_{\alpha}^{\alpha/2}
 \right)^2\right).
\end{align*}
In particular, we have that $|U_{g,h} (u,v)| \leq  2|uv|^{\alpha/2} \int_{0}^{\infty} |g(x) h(x)|^{\alpha/2} dx$. 
\end{lem}
Now we turn our attention to formula  \eqref{asycov}, which presents an explicit expression for the asymptotic covariance matrix $\text{cov}(W)$. In the following we will prove this identity. For the sake of brevity we will only show formula   \eqref{asycov} for $d=1$ and only for the component 
$\text{var}(W_1^{(1)})$ with $k_1=k$ and $r_1=r$. All other identities are in fact easier to prove and we leave them to the reader.

The expression for  $\text{var}(W_1^{(1)})$ for $p \in (-1/2,0)$ and its finiteness have been shown in \cite[Corollary 3.3 and Theorem 4.2]{DI} using methods from distribution theory, so we concentrate on the case $p \in (0,1/2)$. 
For $p\in (0,1)$, we have the relationship 
\begin{align} \label{xp}
|x|^p = a_p^{-1} \int_{\R} \left(1-\exp(ixy) \right) |y|^{-1-p} dy.
\end{align}
which can be shown by substitution $xy=z$ (recall the definition of $a_p$ at \eqref{ap}). 
Note that similarly to \eqref{covFourier} the latter connects
power functions with characteristic functions, which are explicit in the $\alpha$-stable case.  Applying
this formula and using stationarity of the increments $\Delta_{i,k}^{r} X$ we conclude that 
\begin{align*}
\text{cov}\left(
f_p\left( \Delta_{i,k}^{r} X \right) ,
f_p\left( \Delta_{i+l,k}^{r} X \right)   \right) 
= \theta(h_{k,r}, h_{k,r}(\cdot+l))_{p},
\end{align*}
where the quantity $\theta(g,h)_p$ has been introduced at \eqref{powercov}. Since $W(n)_1^{(1)}$ is 
a sum of stationary random variables it remains to prove that $ \theta(h_{k,r}, h_{k,r}(\cdot+l))_{p}$ 
is absolutely summable in $l$ to show  the identity \eqref{asycov}. This is the statement of the next 
lemma.

\begin{lem} \label{lem4}
For $p \in (0,1/2)$ with $p<\alpha/2$ it holds that 
\[
|\theta(h_{k,r}, h_{k,r}(\cdot+l))_{p}| \leq C \rho_l.
\]
In particular, if $k>H+1/\alpha$ we obtain 
$\sum_{l=1}^{\infty} |\theta(h_{k,r}, h_{k,r}(\cdot+l))_{p}| < \infty$.   
\end{lem}
\begin{proof}
The second part of the statement follows directly from Lemma \ref{lem2} and the fact that $(\alpha(H-k) -1)/2 < -1$ when $k>H+1/\alpha$. To show the first part of the statement we will use the inequalities of
Lemma \ref{lem3}. Recalling the definition of $\theta(h_{k,r}, h_{k,r}(\cdot+l))_{p}$ it is sufficient to compute the double integral over the set $(0,\infty)^{2}$ (instead of $\R^2$), which is due to symmetry. 
The domain $(0,\infty)^{2}$ is further decomposed into the regions $(0,1)^2$, $(0,1) \times [1, \infty)$,
$[1, \infty) \times (0,1)$ and $[1, \infty)^2$, and we denote the corresponding integrals by $I_1, I_2,I_3$ and $I_4$, respectively. 

For the integral $I_1$ we use the inequality 
$|U_{g,h} (u,v)| \leq  2|uv|^{\alpha/2} \int_{0}^{\infty} |g(x) h(x)|^{\alpha/2} dx$ of Lemma \ref{lem3} 
to deduce that
\begin{align*}
|I_1| \leq a_p^{-2} \int_{(0,1)^2}    (xy)^{-1-p}|U_{h_{k,r},h_{k,r}(\cdot+l)}(x,y)| dxdy
\leq C \rho_l  \int_{(0,1)^2}    (xy)^{-1-p+\alpha/2} dxdy,
\end{align*}
where the last integral is finite because $p<\alpha/2$.  Applying the main statement of Lemma \ref{lem3} we also conclude the inequality 
\begin{align*}
|I_4| \leq  C \rho_l  \int_{[1,\infty)^2}    (xy)^{-1-p+\alpha/2} 
\exp\left(-2(xy)^{\alpha/2} \left( \|h_{k,r}\|_{\alpha}^{\alpha}
- \rho_l \right) \right) dxdy,
\end{align*}
By Cauchy-Schwarz inequality we have that $\rho_l <  \|h_{k,r}\|_{\alpha}^{\alpha}$. Furthermore,
$\lim_{l\to \infty } \rho_l  = 0$ by Lemma \ref{lem2} and thus, for a given $\epsilon \in (0,1)$, $\rho_l < \epsilon$ for almost all $l \in \N$. Hence, there exists a constant $C>0$ such that 
\[
|I_4| \leq  C \rho_l  \int_{[1,\infty)^2}    (xy)^{-1-p+\alpha/2} 
\exp\left(-2C(xy)^{\alpha/2}   \right) dxdy,
\]
where the latter integral is obviously finite. For the integral $I_2$ we apply Lemma \ref{lem3} once more to obtain
\begin{align*}
|I_2|&\leq C \rho_l \int_{(0,1)\times [1,\infty)} (xy)^{-1-p+\alpha/2} 
\exp\left(-\| h_{k,r}\|_{\alpha}^{\alpha/2} (y^{\alpha/2} - x^{\alpha/2})^2  \right) dxdy \\[1.5 ex]
&\leq C \rho_l \int_{(0,1)\times [1,\infty)} (xy)^{-1-p+\alpha/2} 
\exp\left(-\| h_{k,r}\|_{\alpha}^{\alpha/2} ( y^{\alpha/2} -1)^2  \right) dxdy
\end{align*} 
and the last integral is again finite since $p<\alpha/2$. The term $I_3$ is treated exactly the same way as $I_2$ and we are done. 
\end{proof}
At the end of this subsection we remark that the covariance matrix $\text{cov}(W)$ is a continuous function
in $(\sigma,\alpha, H) \in \R_+ \times (0,2) \times (0,1)$, which follows by Lemma \ref{lem4} and a dominated convergence theorem.

\subsection{Proof of Theorem  \ref{multivariateclt}} \label{sec6.2}

The proof of Theorem \ref{multivariateclt} will be divided into several steps. Some parts of the proof will rely upon asymptotic expansions investigated in \cite{BLP, PT2}.

\subsubsection{Asymptotic decomposition of the statistic $\left(W(n)^{(1)}, W(n)^{(2)} \right)$}  \label{sec6.2.1}

In this section we introduce several approximations of the statistic appearing in Theorem  \ref{multivariateclt}. We start with the asymptotically normal part 
$\left(W(n)^{(1)}, W(n)^{(2)} \right)$. Recalling the notation  \eqref{def-h} we observe the identity
\begin{align} \label{incidentity}
\Delta_{i,k}^{r} X= \int_{\R} h_{k,r} (i-s) dL_s. 
\end{align}
In the first step we introduce the short memory approximation of $\Delta_{i,k}^{r} X$ by truncating the integration region:
\begin{align} \label{incapp}
\Delta_{i,k}^{r} X (m):= \int_{i-m}^{i+m} h_{k,r} (i-s) dL_s. 
\end{align}
Note that the random variables  $(\Delta_{i,k}^{r} X (m))_{i \geq rk}$ are stationary 
and $2m$-dependent, i.e.   $\Delta_{i,k}^{r} X (m)$ and $\Delta_{j,k}^{r} X (m)$ are independent if $|i-j|\geq 2m$. For $f(x)=|x|^p$ with $p \in (0,1/2)$ and $p<\alpha/2$, or $f(x)=\cos(tx)$ 
we introduce the notation
\begin{align} \label{Wnm}
W(n,m)_j^{(1)}:= \frac{1}{\sqrt{n}} \sum_{i=r_jk_j}^n \left\{ f_p\left( \Delta_{i,k_j}^{r_j} X(m) \right)
-\E \left[ f_p\left( \Delta_{i,k_j}^{r_j} X(m) \right) \right] \right\} \\[1.5 ex]
W(n,m)_j^{(2)}:= \frac{1}{\sqrt{n}} \sum_{i=k}^n \left\{ \psi_{t_j}\left( \Delta_{i,k_j} X(m) \right)
-\E \left[ \psi_{t_j}\left( \Delta_{i,k_j} X(m) \right) \right] \right\} \nonumber
\end{align}
For the function $f_{-p}(x)=|x|^{-p}$ with $p \in (0,1/2)$ we set $f_{-p}^{\epsilon}(x)=|x|^{-p} 1_{\{|x|>\epsilon\}}$ and note that the latter is a bounded function. In this setting we define  
\begin{align} \label{Wnme}
W(n,m,\epsilon)_j^{(1)}:= \frac{1}{\sqrt{n}} \sum_{i=rk}^n \left\{ f_{-p}^{\epsilon}
\left( \Delta_{i,k}^{r} X(m) \right)
-\E \left[ f_{-p}^{\epsilon}\left( \Delta_{i,k}^{r} X(m) \right) \right] \right\}.
\end{align}
In \cite[Section 5.4]{BLP} it has been shown that the convergence
\begin{align}  \label{bill1}
\lim_{m \to \infty} \limsup_{n \to \infty} \E\left[\left(W(n,m)_j^{(1)} - W(n)_j^{(1)} \right)^2\right]=0
\end{align}
holds. On the other hand, since the functions $\psi_{t_j}$ and $f_{-p}^{\epsilon}$ are bounded, we obtain the convergence
\begin{align}  \label{bill2}
\lim_{m \to \infty} \limsup_{n \to \infty} \E\left[\left(W(n,m)_j^{(2)} - W(n)_j^{(2)} \right)^2\right]=0, \\[1.5 ex]
\lim_{m \to \infty} \limsup_{n \to \infty} \E\left[\left(W(n,m, \epsilon)_j^{(1)} - 
W(n,\epsilon)_j^{(1)} \right)^2\right]=0 \nonumber
\end{align}
from \cite{PT2}. Here $W(n,\epsilon)_j^{(1)}$ is the original statistic defined at \eqref{multstat} associated with the function $f_{-p}^{\epsilon}$.

\subsubsection{Asymptotic decomposition of the statistic $\left(S(n)^{(1)}, S(n)^{(2)} \right)$}  \label{sec6.2.1b}

In this subsection we derive an asymptotic expansion for the statistic 
$\left(S(n)^{(1)}, S(n)^{(2)} \right)$. The main ideas originate from the work \cite{BLP} and we will adapt their principles to our setting. The following estimates and decomposition have been treated in the case of power variation with $p \in (0,1/2)$, $p<\alpha/2$, in \cite{BLP}, so we will rather concentrate 
on the functions $f_{-p}$, $p \in (0,1/2)$, and $\psi_t$.   

All expansions are valid componentwise, so we may assume that $d=1$. We recall the notation 
introduced at \eqref{Phi12}. For a symmetric $\alpha$-stable random variable $Y$ with scaling parameter $\rho>0$ and a measurable function $f: \R \to \R$, we introduce the function 
\begin{align} \label{Phirho}
\Phi_{\rho}(f)(x):= \E[f(Y+x)] - \E[f(Y)], \qquad x \in \R,
\end{align}
whenever the latter is finite. In the following we will derive various estimates for $\Phi_{\rho}(f_{-p})(x)$
with  $p \in (0,1/2)$. First of all, using the identity \cite[Eq. (18)]{DI} we obtain the representation 
\begin{align} \label{pminusid}
\Phi_{\rho}(f_{-p})(x) = a_p^{-1} \int_{\R} \left(1-\cos(xy) \right) 
\exp(-|\rho y|^{\alpha}) |y|^{-1+p} dy.
\end{align}  
This identity implies the following result.

\begin{lem} \label{lemPhi}
Assume that $\rho, \rho_1, \rho_2> \epsilon>0$. Then there exists a constant $C_{\epsilon}>0$ such that the following inequalities hold:
\begin{align*}
&|\Phi_{\rho}(f_{-p})(x)| \leq C_{\epsilon} ( 1 \wedge x^2), \qquad |\Phi_{\rho}(f_{-p})^{(v)}(x)| \leq C_{\epsilon} \quad \text{for } v=1,2, \\[1.5 ex]
&|\Phi_{\rho}(f_{-p})(x) - \Phi_{\rho}(f_{-p})(y)| \leq C_{\epsilon}  \left((1 \wedge |x| + 
1 \wedge |y|)|x-y| 1_{\{|x-y| \leq1 \}} + 1_{\{|x-y| >1 \}} \right), \\[1.5 ex]
&|\Phi_{\rho_1}(f_{-p})(x) - \Phi_{\rho_2}(f_{-p})(x)| \leq C_{\epsilon} |\rho_2^{\alpha} - \rho_1^{\alpha}|,
\\[1.5 ex]
& \int_{0}^{x}  \int_{0}^{y} \Phi_{\rho}(f_{-p})(a+z+w)| dz dw  \leq C_{\epsilon} (1 \wedge x)
(1 \wedge y) \qquad \text{for any } x,y>0,~ a\in \R,
\end{align*}
where $\Phi_{\rho}(f_{-p})^{(v)}$ denotes the $v$th derivative of $\Phi_{\rho}(f_{-p})$.
\end{lem}
\begin{proof}
Note that the function $f_{-p}$ is even and hence $\Phi_{\rho}(f_{-p})(0)= \Phi_{\rho}(f_{-p})^{(1)}(0)=0$. Using the identity  \eqref{pminusid} we immediately see that 
$|\Phi_{\rho}(f_{-p})^{(v)}(x)| \leq C_{\epsilon}$ for $v=0,1,2$. Thus, we obtain the first two inequalities.
By the same arguments we get $|\Phi_{\rho}(f_{-p})^{(1)}(x)| \leq C_{\epsilon} ( 1 \wedge |x|)$. Observing the identity 
\[
|\Phi_{\rho}(f_{-p})(x) - \Phi_{\rho}(f_{-p})(y)| = \left| \int_y^x \Phi_{\rho}(f_{-p})^{(1)}(u) du\right|
\] 
we readily deduce the third inequality. The fourth inequality follows immediately from 
\eqref{pminusid} and the mean value theorem. The last statement is a straightforward consequence 
of the first three inequalities of Lemma \ref{lemPhi}.  
\end{proof}
It is important to note that the result of Lemma \ref{lemPhi} remains valid for the function 
$\Phi_{\rho}(\psi_t)$. In this case it is a consequence of the fact the $\psi_t$ is a bounded and even function. 

In the next step we present some decompositions, which have been investigated in \cite{BLP}. For any 
fixed $r$ and k, and the function $f=f_p, f_{-p}$, $p \in (0,1/2)$, or $\psi_t$, we define the random variable
\[
S(f)_n= n^{- 1/(1+\alpha(k-H))} \sum_{i=rk}^n \left\{ f\left( \Delta_{i,k}^{r} X \right) - 
\E\left[ f\left( \Delta_{i,k}^{r} X \right) \right] \right\} =:  \sum_{i=rk}^n V_i^n.
\] 
We also introduce the $\sigma$-algebras 
\[
\mathcal{G}_s:= \sigma\left(L_v - L_u:~ v,u \leq s \right), \qquad 
\mathcal{G}_s^1:= \sigma\left(L_v - L_u:~ s \leq v,u \leq s+1 \right), 
\]
and note that $(\mathcal{G}_s^1)_{s \in \R}$ is not a filtration. Now, we introduce the notation
\begin{align*}
& R_i^n:= \sum_{j=1}^n \zeta_{i,j}^n, \qquad Q_i^n:=  \sum_{j=1}^n \E[V_i^n| \mathcal{G}_{i-j}^1], 
\\[1.5 ex]
& \text{where } \zeta_{i,j}^n:= \E[V_i^n| \mathcal{G}_{i-j+1}] - \E[V_i^n| \mathcal{G}_{i-j}]
- \E[V_i^n| \mathcal{G}_{i-j}^1].
\end{align*}
Finally, we observe the decomposition  
\begin{align} \label{barSf}
S(f)_n &= \sum_{i=rk}^n R_i^n + \left(- \overline{S}(f)_n + \sum_{i=rk}^n Q_i^n \right)
+ \overline{S}(f)_n, \\[1.5 ex]
\overline{S}(f)_n &:= n^{- 1/(1+\alpha(k-H))}\sum_{i=rk}^n 
\left\{ \overline{\Phi}(f)(L_i-L_{i-1}) - \E[\overline{\Phi}(f)(L_i-L_{i-1})] \right\}, \nonumber
\end{align}
where $\overline{\Phi}(f)(x):= \sum_{j=1}^{\infty} \Phi_{\rho} (f) \left(h_{k,r}(j)x\right)$ with $\rho= \sigma
\|h_{k,r}\|_{\alpha}$. Note that $\overline{S}(f)_n$ is a sum of  i.i.d random variables.
For $f=f_p$ with $p \in (0,1/2)$, $p <\alpha/2$ and under assumptions of Theorem
\ref{multivariateclt}, the convergence 
\begin{align} \label{RQ}
\sum_{i=rk}^n R_i^n \toop 0 \qquad \text{and} \qquad - \overline{S}(f)_n + \sum_{i=rk}^n Q_i^n \toop 0
\qquad \text{as } n\to \infty
\end{align}
has been shown in \cite{BLP} (cf. eqs. (5.30), (5.31) and (5.38) therein). The proof of these convergence results follows
from a number of estimates on the function $\Phi_{\rho}(f_{p})$, $p \in (0,1/2)$, which are stated in
\cite[eqs. (5.14)-(5.18) and Lemma 5.8]{BLP}. But according to Lemma \ref{lemPhi} the same estimates hold also for   $\Phi_{\rho}(f_{-p})$, $p \in (0,1/2)$, and $\Phi_{\rho}(\psi_t)$ (in fact, the latter estimates are stronger). Consequently, the convergence at \eqref{RQ} also holds for the cases
$f=f_{-p}$ and $f=\psi_t$ and we deduce that 
\begin{align} \label{neg}
S(f)_n - \overline{S}(f)_n \toop 0 \qquad \text{for } f=f_p, f_{-p} \text{ or } \psi_t.
\end{align}

\subsubsection{A limit theorem for the approximations}  \label{sec6.2.2}

Recalling the notation introduced in \eqref{Phi12} and \eqref{barPhi12} we obtain the identities 
\begin{align} \label{lineS}
\overline{S}(n)^{(1)}_j &:=  \overline{S}(f_p)_{n,j} =
n^{- 1/(1+\alpha(k-H))} \sum_{i=r_jk}^n 
\left\{ \overline{\Phi}_j^{(1)} (L_i-L_{i-1}) - \E[\overline{\Phi}_j^{(1)} (L_i-L_{i-1})] \right\}, \nonumber
\\[1.5 ex]
\overline{S}(n)^{(2)}_j &:=  \overline{S}(\psi_{t_j})_{n} =
n^{- 1/(1+\alpha(k-H))} \sum_{i=r_jk}^n 
\left\{ \overline{\Phi}_j^{(2)} (L_i-L_{i-1}) - \E[\overline{\Phi}_j^{(2)} (L_i-L_{i-1})] \right\}, 
\end{align}
where $p \in (-1/2,1/2) \setminus \{0\}$ and the statistic $ \overline{S}(f_p)_{n,j}$ is defined 
as in \eqref{barSf} using the
parameters $r_j$ and $k$.  
As a consequence of \eqref{bill1}, \eqref{bill2} and \eqref{neg} 
it is now sufficient to show a weak limit theorem 
for the statistic $$(W(n,m)^{(1)}, W(n,m)^{(2)}, \overline{S}(n)^{(1)}, \overline{S}(n)^{(2)})$$ (resp. 
$(W(n,m, \epsilon)^{(1)}, W(n,m)^{(2)}, \overline{S}(n)^{(1)}, 
\overline{S}(n)^{(2)})$) when $p \in (0,1/2)$ and 
$p<\alpha/2$ (resp. $-p \in (0,1/2)$) as $n \to \infty$ and then $m\to \infty$. 

In order to prove this convergence we recall the results of \cite{RG} adapted to our setting.  Let 
$(Y_i^{(1)})_{i \geq 1}$ and $(Y_i^{(2)})_{i \geq 1}$ be i.i.d sequences of centred random variables 
of dimensions $d_1$ and $d_2$ respectively, which are not necessarily independent. 
Define the statistics
\[
Z_n^{(1)} = \frac{1}{\sqrt{n}} \sum_{i=1}^n Y_i^{(1)}, \qquad 
Z_n^{(2)} = n^{-1/\beta}\sum_{i=1}^n Y_i^{(2)} \qquad \text{with } \beta \in (1,2).
\]
Assume now that $Z_n^{(1)} \schw Z^{(1)}$ where $Z^{(1)}$ is a $d_1$-dimensional centred normal distribution  and assume that each coordinate $Y_{1,j}^{(2)}$, $1\leq j\leq d_2$, is in the domain of attraction of a $\beta$-stable random variables, i.e.
\[
\lim_{x \to +\infty} x^{\beta} \mathbb{P}(Y_{1,j}^{(2)} > x) = b_j^+ \qquad \text{and}
\qquad \lim_{x \to -\infty} |x|^{\beta} \mathbb{P}(Y_{1,j}^{(2)} < x) = b_j^-.
\]
Assume moreover that 
there exists a measure $\overline{\nu}$ such that for all sets $A\in \mathcal{B}(\R^{d_2})$ bounded away from $0$ with $\overline{\nu}(\partial A)=0$ it holds:
\[
\lim_{n\to \infty} n \mathbb{P}(n^{-1/\beta}Y_i^{(2)} \in A)= \overline{\nu}(A).
\]
Then we obtain the joint convergence 
\begin{align} \label{resnick}
\left(Z_n^{(1)}, Z_n^{(2)} \right) \schw \left(Z^{(1)}, Z^{(2)} \right),
\end{align}
where $Z^{(1)}$ and $Z^{(2)}$ are necessarily independent, and the law of $Z^{(2)}$ is determined by the
L\'evy measure $\overline{\nu}$. Indeed this result is a direct consequence of \cite[Theorems 3 and 4]{RG} and their direct extension from bivariate to $(d_1+d_2)$-dimensional setting. 
  
Next, we apply the weak convergence at \eqref{resnick} to our framework. Notice first that the statistics
$W(n,m)^{(1)}$, $W(n,m)^{(2)}$ and $W(n,m, \epsilon)^{(1)}$ are sums of $2m$-dependent random variables, but this setting can be reduced to sums of i.i.d random variables by the classical Bernstein's blocking technique. Hence, the theory of \cite{RG} also applies in this case.   

For the sake of brevity we apply the convergence at \eqref{resnick} only for the statistic 
$(W(n,m)^{(1)}, W(n,m)^{(2)}, \overline{S}(n)^{(1)}, \overline{S}(n)^{(2)})$. We set 
\[
Z_{n,m}^{(1)}= \left(W(n,m)^{(1)}, W(n,m)^{(2)} \right) \qquad \text{and} \qquad
Z_n^{(2)} = \left( \overline{S}(n)^{(1)}, \overline{S}(n)^{(2)} \right),
\] 
and define $\beta=(1+\alpha(k-H))$.
By the standard central limit theorem for sums of stationary $2m$-dependent random variables we deduce the convergence
\[
Z_{n,m}^{(1)} \schw Z_{m}^{(1)} \sim \mathcal N_{2d}(0, \Sigma_m) \qquad \text{as } n\to \infty,
\]
where the asymptotic covariance matrix $\Sigma_m$ is defined by 
\begin{align*}
\Sigma_m^{ij} &=   \sum_{l = -2m+1}^{2m-1} \text{cov}
\left( f_p\left( \Delta_{r_i k_i,k_i}^{r_i} X(m) \right), f_p\left( \Delta_{r_i k_i+l,k_j}^{r_j} X(m)\right)
\right), \qquad 1\leq i,j \leq d,\\
\Sigma_m^{ij} &=   \sum_{l = -2m+1}^{2m-1} \text{cov}
\left( f_p\left( \Delta_{r_i k_i,k_i}^{r_i} X(m) \right), \psi_{t_j} \left( \Delta_{r_i k_i+l,k_j} X(m)\right)
\right) , \qquad d+1\leq i+d,j \leq 2d,  \nonumber \\
\Sigma_m^{ij}&= \sum_{l = -2m+1}^{2m-1}   \text{cov}
\left( \psi_{t_i} \left( \Delta_{ k_i,k_i} X(m)\right), \psi_{t_j} \left( \Delta_{k_i+l,k_j} X(m)\right)
\right), \qquad d+1\leq i,j \leq 2d .  \nonumber
\end{align*}
In the next step we treat the statistic $Z_n^{(2)}$. Recalling the definition 
 at \eqref{lineS}, and the tail convergence of 
\eqref{cj1} and \eqref{cj2}, we conclude that the limits of   $\overline{S}(n)^{(1)}$ and $\overline{S}(n)^{(2)}$ must be independent since $b_j^-=0$ for $1\leq j \leq d$ and $b_j^+=0$ for 
$d+1\leq j \leq 2d$. Furthermore, \eqref{nul} readily implies the convergence
\[
Z_n^{(2)} \schw \left(S^{(1)}, S^{(2)} \right) \qquad \text{as } n\to \infty,
\] 
where the vector $\left(S^{(1)}, S^{(2)} \right)$ has been introduced in Theorem \ref{multivariateclt}.

Finally, we will prove that the covariance matrix $ \Sigma_m$ converges as $m \to \infty$. 
In the following we write $\|Y\|_{\mathbb{L}^2}$ for $\E[Y^2]^{1/2}$ for any square integrable random variable $Y$.  
For $m_1,m_2 \in \N$ and $1 \leq j \leq d$ observe the decomposition
\begin{align} \label{sigmaconv}
&| (\Sigma_{m_1}^{jj})^{1/2} - (\Sigma_{m_2}^{jj})^{1/2}| = \lim_{n \to \infty} 
|\|W(n,m_1)_j^{(1)}\|_{\mathbb{L}^2} - \|W(n,m_2)_j^{(1)}\|_{\mathbb{L}^2}| \\
& \leq \limsup_{n \to \infty} \left( \|W(n,m_1)_j^{(1)} - W(n)_j^{(1)}\|_{\mathbb{L}^2} 
+ \|W(n,m_2)_j^{(1)} - W(n)_j^{(1)}\|_{\mathbb{L}^2}  \right) \nonumber
\end{align}
and the latter converges to $0$ as $m_1, m_2 \to \infty$ due to \eqref{bill1}. Hence, 
$(\Sigma_{m}^{jj})_{m \geq 1}$ is a Cauchy sequence and thus it converges. Since $\text{var}(W(n)_j^{(1)}) \to \text{var}(W_j^{(1)})$ we must have that 
\[
\lim_{m \to \infty} \Sigma_{m}^{jj} = \text{var}(W_j^{(1)}). 
\]
The same argument applies to $\Sigma_{m}^{jj}$ for $d+1\leq j \leq 2d$ and also to covariances $\Sigma_{m}^{ij}$ due to polarisation identity. 

Summarising the results of Sections \ref{sec6.2.1}-\ref{sec6.2.2} we obtain the weak limit theorem
\begin{align*} 
\left(W(n)^{(1)}, W(n)^{(2)}, S(n)^{(1)}, S(n)^{(2)} \right) \schw 
\left(W^{(1)}, W^{(2)}, S^{(1)}, S^{(2)} \right),  
\end{align*}
for $p \in (0,1/2)$ and $p<\alpha/2$, as claimed in \eqref{mclt}. Similarly, for $-p \in (0,1/2)$ we have also obtained the convergence
\begin{align*} 
\left(W(n,\epsilon)^{(1)}, W(n)^{(2)}, S(n)^{(1)}, S(n)^{(2)} \right) \schw 
\left(W(\epsilon)^{(1)}, W^{(2)}, S^{(1)}, S^{(2)} \right),  
\end{align*}
for any $\epsilon>0$. Here the limit $(W(\epsilon)^{(1)}, W^{(2)}, S^{(1)}, S^{(2)})$ is defined
as in Theorem \ref{multivariateclt}, where the function $f_{-p}$ is replaced by $f_{-p}^{\epsilon}$. In 
order to prove the original theorem for  $-p \in (0,1/2)$ we need to let $\epsilon \to 0$, which is the subject of the next subsection.

\subsubsection{Letting $\epsilon \to 0$}  \label{sec6.2.3}

For simplicity we may assume that $d=1$ and $r_1=r, ~k_1=k$. In the first step we will show that 
\begin{align*}
\lim_{\epsilon \to 0} \limsup_{n \to \infty} \E\left[\left(W(n, \epsilon)^{(1)} - 
W(n)^{(1)} \right)^2\right]=0.
\end{align*} 
We define the function $\bar{f}_{-p}^{\epsilon}= f_{-p} - f_{-p}^{\epsilon}$, $p \in (0,1/2)$. Notice that 
$\text{supp}(\bar{f}_{-p}^{\epsilon})=[-\epsilon, \epsilon]$ and 
$|\mathfrak{F}^{-1}(\bar{f}_{-p}^{\epsilon})| \leq C \epsilon$ for all $x \in \R$. Applying the formula \eqref{covFourier} we conclude that 
\begin{align*}
\left|\text{cov}\left( 
\bar{f}_{-p}^{\epsilon}\left( \Delta_{i,k}^{r} X \right) ,
\bar{f}_{-p}^{\epsilon}\left( \Delta_{i+l,k}^{r} X \right)   \right) \right | \leq C \epsilon^2 \int_{\R^2}
|U_{h_{k,r},h_{k,r}(\cdot+l)}(x,y)| dxdy.
\end{align*}
In \cite[Lemma 3.4]{PTA} it has been proved that the inequality $\int_{\R^2}
|U_{h_{k,r},h_{k,r}(\cdot+l)}(x,y)| dxdy \leq C \rho_l$ holds (in fact, the proof is the same as for Lemma \ref{lem4}). Hence, we conclude by Lemma \ref{lem2} and the condition $k> H+1/\alpha$ 
\begin{align} \label{covest} 
\left|\text{cov}\left( 
\bar{f}_{-p}^{\epsilon}\left( \Delta_{i,k}^{r} X \right) ,
\bar{f}_{-p}^{\epsilon}\left( \Delta_{i+l,k}^{r} X \right)   \right) \right | \leq C \epsilon^2
l^{(\alpha(H-k) -1)/2} 
\end{align}
Since $(\alpha(H-k) -1)/2<-1$ when $k> H+1/\alpha$ we readily deduce the estimate
\[
\limsup_{n \to \infty} \E\left[\left(W(n, \epsilon)^{(1)} - 
W(n)^{(1)} \right)^2\right] \leq C \epsilon^2
\]
and the first statement follows. 

Now, we are left to proving weak convergence for the vector $(W(\epsilon)^{(1)}, W^{(2)})$ as $\epsilon \to 0$. This random variable is bivariate normal with mean $0$. Hence, it suffices to show that the covariance matrix converges. But this follows by setting $\epsilon=1/N$ and applying a Cauchy sequence argument as presented in \eqref{sigmaconv}. Thus, the proof of Theorem \ref{multivariateclt} is complete. \qed

\subsection{Proof of Theorem \ref{th3}} 

Part (i) of Theorem \ref{th3} follows from Theorem \ref{multivariateclt} applied to the setting $d=2$,
$p \in (0,1/2)$, $k_j=k \geq 2$ (and hence $k>H+1/\alpha$), and the classical delta method. In fact, we only use the central limit theorem part of   Theorem \ref{multivariateclt}.

Part (ii) of Theorem \ref{th3} is slightly more involved. We start with the identity ($t>0$)
\begin{align*}
\varphi_{\text{high}}(t; \widehat{H},k)_n
- \varphi_{\text{high}}(t; H,k)_n 
= \frac{1}{n}  \sum_{i=k}^n \left\{\cos( tn^{\widehat{H}} \Delta_{i,k}^{n} X ) - 
\cos( tn^H \Delta_{i,k}^{n} X ) \right \},
\end{align*}
where we use the short notation $\widehat{H} = \widehat{H}_{\text{high}} (p,k)_n$. Setting 
$M_n = n^{\widehat{H} - H}$ and using the inequality 
$|\cos(y)-\cos(x)+(y-x) \sin(x)| \leq C|y-x|^{\alpha'}$ for some $\alpha' \in (1, \alpha)$,
we conclude that 
\begin{align*}
&\varphi_{\text{high}}(t; \widehat{H},k)_n
- \varphi_{\text{high}}(t; H,k)_n =  -\frac{t(M_n-1)}{n}  \sum_{i=k}^n (n^H \Delta_{i,k}^{n} X)
\sin( tn^H \Delta_{i,k}^{n} X ) +R_n, \\[1.5 ex] 
& \text{where } 
|R_n| \leq C \frac{|M_n-1|^{\alpha'}}{n} \sum_{i=k}^n |n^H \Delta_{i,k}^{n} X|^{\alpha'}.
\end{align*}
We observe that $\sqrt{n}(\widehat{H} - H)$ is asymptotically normal, which follows by a delta method from  Theorem \ref{multivariateclt} (take $d=2$ and use the convergence in distribution $W(n)^{(1)} \schw W^{(1)} $). By the mean value theorem we obtain that 
\[
\sqrt{n} (\log n)^{-1} (M_n-1) = \sqrt{n}(\widehat{H} - H) + o_{\mathbb{P}}(1). 
\]
Hence, recalling that $\alpha' \in (1, \alpha)$,  we deduce by Birkhoff's ergodic theorem
\begin{align} \label{asyexp}
\sqrt{n} (\log n)^{-1}
\left(\varphi_{\text{high}}(t; \widehat{H},k)_n - \varphi_{\text{high}}(t; H,k)_n \right) = 
\sqrt{n}(\widehat{H} - H) t \varphi'(t;k)+ o_{\mathbb{P}}(1),
\end{align}
where we used the identity $t \varphi'(t;k)= - \E[n^H \Delta_{i,k}^{n} X
\sin( tn^H \Delta_{i,k}^{n} X )]$. Finally, we note that 
\[
\varphi_{\text{high}}(t; H,k)_n - \varphi(t;k) = O_{\mathbb{P}} (n^{-1/2}),
\]
which follows from Theorem \ref{multivariateclt}. Hence, observing the identities \eqref{G} and 
\eqref{alphasigma}, we obtain the statement of Theorem \ref{th3}(ii) by applying the delta method to 
Theorem \ref{multivariateclt}. \qed

\subsection{Proof of Theorem \ref{th4}} 
First of all, we note that $\delta:= \alpha^{-1} - \lfloor \alpha^{-1} \rfloor \in (0,1)$ since $ \alpha^{-1} 
\not \in \N$. Setting $\delta':= \min\{\delta, 1- \delta\}/2>0$ we conclude 
that 
\begin{align*}
\mathbb{P} \left(\widehat{k}_{\text{low}} \not = 2+  \lfloor \alpha^{-1} \rfloor \right)
\leq \mathbb{P} \left(| \widehat{\alpha}_{\text{low}}^{0} (t_1,t_2)_n^{-1} - \alpha^{-1}|>\delta' \right)
\to 0,
\end{align*}
because $\widehat{\alpha}_{\text{low}}^{0} (t_1,t_2)_n \toop \alpha$ and $\alpha>0$. This implies
the convergence $\widehat{k}_{\text{low}} \toas 2+  \lfloor \alpha^{-1} \rfloor$. Thus, it suffices to prove the asymptotic results of Theorem \ref{th4} when $\widehat{k}_{\text{low}}$  is replaced by $2+  \lfloor \alpha^{-1} \rfloor$. Now, notice that $k= 2+  \lfloor \alpha^{-1} \rfloor$ automatically satisfies the condition $k>H+1/\alpha$
since $H \in (0,1)$. This guarantees that the statistic $(W(n)^{(1)}, W(n)^{(2)})$ defined at \eqref{multstat} is in the domain of attraction of the central limit theorem. Hence,  
Theorem \ref{th4}(i) follows directly by the delta method from Theorem \ref{multivariateclt} (cf. proof
of Theorem \ref{th3}(i)).  \qed

\subsection{Proof of Proposition \ref{prop1}}
Proposition  \ref{prop1} is shown by exactly the same arguments as Theorem \ref{th3}. \qed

\subsection{Proof of Theorem \ref{th5}}
Recall that $\alpha^{-1} \in \N$. Hence, we have 
\begin{align*}
\mathbb{P} \left(\widehat{k}_{\text{low}} \not  \in \{1+   \alpha^{-1}, 2+   \alpha^{-1}  \} \right)
\leq \mathbb{P} \left(| \widehat{\alpha}_{\text{low}}^{0} (t_1,t_2)_n^{-1} - \alpha^{-1}|>1 \right)
\to 0,
\end{align*}
because $\widehat{\alpha}_{\text{low}}^{0} (t_1,t_2)_n \toop \alpha$ and $\alpha>0$.  
Note that $k \in \{1+   \alpha^{-1}, 2+   \alpha^{-1} \}$ 
satisfies the condition $k>H +  \alpha^{-1}$, which guarantees the validity of a central limit theorem
for the statistic $(W(n)^{(1)}, W(n)^{(2)})$ defined at \eqref{multstat}. 

We introduce the notation
\[
T_{\text{low}}(\widehat{k}_{\text{\rm low}}, n):=\sqrt{n} \left( 
\begin{array} {c}
\widetilde{\sigma}_{\text{\rm low}} (\widehat{k}_{\text{\rm low}},t_1,t_2)_n - \sigma \\
\widetilde{\alpha}_{\text{\rm low}} (\widehat{k}_{\text{\rm low}},t_1,t_2)_n - \alpha \\
\widehat{H}_{\text{\rm low}} (-p,\widehat{k}_{\text{\rm low}})_n - H 
\end{array}
\right)
\]
and 
\[
a_n := 
\begin{cases}
\sqrt{n}: & \text{if } H<1-\alpha^{-1} \\
n^{1 - 1/(1+\alpha(1-H))}: & \text{if } H>1-\alpha^{-1}
\end{cases}
\]
We set $U_n=a_n(\widehat{\alpha}_{\text{\rm low}}^{0} (t_1,t_2)_n - \alpha)$, $A=(a_1,b_1) \times (a_2,b_2) \times (a_3,b_3)$ and observe the decomposition
\begin{align*}
\mathbb{P}(T_{\text{low}}(\widehat{k}_{\text{\rm low}}, n) \in A) &= 
\mathbb{P} \left(T_{\text{low}}(1+   \alpha^{-1}, n) \in A, ~
\widehat{\alpha}_{\text{low}}^{0} (t_1,t_2)_n^{-1} - \alpha^{-1}<0 \right) \\[1.5 ex]
&+ \mathbb{P} \left(T_{\text{low}}(2+   \alpha^{-1}, n) \in A, ~
\widehat{\alpha}_{\text{low}}^{0} (t_1,t_2)_n^{-1} - \alpha^{-1} \geq 0 \right) + o(1) \\[1.5 ex]
&=\mathbb{P} \left(T_{\text{low}}(1+   \alpha^{-1}, n) \in A, ~
U_n>0 \right) \\[1.5 ex]
&+ \mathbb{P} \left(T_{\text{low}}(2+   \alpha^{-1}, n) \in A, ~
U_n \leq 0 \right) + o(1).
\end{align*}
Applying Theorem \ref{multivariateclt} and Proposition \ref{prop1}, and using the same arguments 
as in the proof of Theorem \ref{th3}, we thus conclude the convergence
\begin{align*}
\lim_{n \to \infty} \mathbb{P}(T_{\text{low}}(\widehat{k}_{\text{\rm low}}, n) \in A) &= 
\mathbb{P} \left(B_{\text{\rm low}}^{\text{\rm nor}} (-p, 1+\alpha^{-1}) 
\in A,~ B_{\text{\rm low}}\left(-p,1 \right)_2>0  \right) \\[1.5 ex]
&+ \mathbb{P} \left(B_{\text{\rm low}}^{\text{\rm nor}} (-p, 2+\alpha^{-1}) 
\in A,~ B_{\text{\rm low}}\left(-p,1 \right)_2<0  \right),
\end{align*}
where $B_{\text{\rm low}}(-p,1 )= B_{\text{\rm low}}^{\text{nor}}(-p,1 )$ if $H<1-\alpha^{-1}$, and $B_{\text{\rm low}}(-p,1 )= B_{\text{\rm low}}^{\text{sta}}(-p,1 )$ if $H>1-\alpha^{-1}$. 
Hence, we immediately obtain the assertion of Theorem \ref{th5}. \qed

\subsection{Proof of Theorem \ref{th6}}
As in the proof of Theorem \ref{th4} we conclude that $\widehat{k}_{\text{high}} \toas 2 + \lfloor \alpha^{-1} \rfloor$. On the other hand, similarly to \eqref{asyexp}, we obtain the asymptotic expansion
\begin{align*}
\sqrt{n} (\log n)^{-1}
\left(V_{\text{high}}(f_{-p}, \widehat{H}_{\text{high}} (-p)_n)_n - V_{\text{high}}(f_{-p})_n \right) = 
-\sqrt{n}(\widehat{H} - H)p  m_{p,k}+ o_{\mathbb{P}}(1),
\end{align*}
for any $p \in (0,1/2)$. Hence, the assertion of Theorem \ref{th6} follows the delta method and Theorem \ref{multivariateclt} (cf. the proof of Theorem \ref{th3}). \qed

\subsection{Proof of Theorem \ref{th7}}
The results of Theorem \ref{th7} follow by the same methods as presented in the proof of Theorem \ref{th5}. \qed

\end{document}